\documentclass[11pt]{article}

\usepackage[margin=1in]{geometry}
\usepackage{amsmath,amssymb,amsthm}
\usepackage{algorithm}
\usepackage{algpseudocode}
\usepackage{graphicx}
\usepackage{hyperref}
\usepackage{cleveref}
\usepackage{booktabs}
\usepackage{tikz}
\usepackage{xcolor}
\usepackage{enumitem}
\usepackage{pgfplots}
\usepackage{float}
\usepackage{placeins}
\pgfplotsset{compat=1.18}
\usetikzlibrary{positioning,calc,arrows.meta,patterns,decorations.pathreplacing}

\newtheorem{theorem}{Theorem}[section]
\newtheorem{lemma}[theorem]{Lemma}
\newtheorem{proposition}[theorem]{Proposition}
\newtheorem{corollary}[theorem]{Corollary}
\newtheorem{definition}[theorem]{Definition}
\newtheorem{remark}[theorem]{Remark}
\newtheorem{example}[theorem]{Example}

\crefname{definition}{Definition}{Definitions}
\Crefname{definition}{Definition}{Definitions}
\crefname{remark}{Remark}{Remarks}
\Crefname{remark}{Remark}{Remarks}
\crefname{lemma}{Lemma}{Lemmas}
\Crefname{lemma}{Lemma}{Lemmas}
\crefname{proposition}{Proposition}{Propositions}
\Crefname{proposition}{Proposition}{Propositions}
\crefname{corollary}{Corollary}{Corollaries}
\Crefname{corollary}{Corollary}{Corollaries}
\crefname{example}{Example}{Examples}
\Crefname{example}{Example}{Examples}

\newcommand{\FF}{\mathcal{F}}

\newcommand{\NN}{\mathcal{N}}

\newcommand{\im}{\mathrm{im}}
\newcommand{\rk}{\mathrm{rk}}

\newcommand{\R}{\mathbb{R}}

\title{%
  Incremental Sheaf Cohomology on Cellular Complexes:\\
  $O(1)$-in-$n$ Lazy Edit Processing under Bounded Local Geometry
}

\author{
  Jason L.\ Volk \\
  Invariant Research \\
  \texttt{jason@invariant.pro}
}

\date{}

\begin{document}
\maketitle

\begin{abstract}
We present an algorithmic framework for incremental maintenance of
first sheaf cohomology $H^1(X; \FF)$ on dynamically evolving
1-dimensional cellular complexes equipped with finite-dimensional
cellular sheaves. Following the standard convention in
dynamic-algorithm analysis, we distinguish \emph{update} time
(the cost to ingest one edit) from \emph{query} time (the cost to
report $\dim H^1$ on demand). The classical computation requires
$O(n^3)$ time; when the complex evolves with a stream of $m$
edits, full recomputation after each edit costs $O(m n^3)$.

Under a \emph{bounded local geometry}
assumption (bounded cell size $v_{\max}$, bounded stalk dimension
$d$, and bounded nerve degree $D$), each edit (vertex insertion,
edge insertion, restriction map update) affects only a bounded set
of local coboundary blocks. The algorithm processes lazy
streaming edits in $O(d^2)$ \emph{update} time, independent of
the total complex size $n$, deferring local eigensolves and
global assembly to synchronization points (\textsc{Flush}).
At synchronization, the \emph{query} cost is $O(n)$ in the
implemented full-traversal assembly path, or $O(1)$-in-$n$ with
an incremental \v{C}ech certificate (described but not yet
implemented). In both paths, the maintained state agrees exactly
with batch recomputation at every synchronization point;
we observe zero measured drift through $V = 5 \times 10^6$.

We also give an amortized
$O(|E|)$ streaming construction for the cellular decomposition,
prove that the nerve degree $D$ is bounded for graph families
with bounded maximum degree (\Cref{prop:bounded-D}), and
establish a degree-independent bound
$D \leq v_{\max} \cdot (K - 1)$ via a vertex multiplicity cap
$\mu(v) \leq K$ enforced by the construction
(\Cref{cor:mu-cap-D}), guaranteeing bounded local geometry
for all graph topologies including scale-free networks. The
invariant the construction enforces is the multiplicity cap
itself ($\mu_{\max} = 3$ in all experiments); together with the
cell-size bound $v_{\max}$ it governs the per-edit update cost,
while the nerve degree, which can be large on scale-free
graphs, enters only the deferred global assembly
(\Cref{rem:nerve-pressure}).
We also present an adversarial algebraic-RAM barrier arguing that
unpartitioned non-trivial sheaves ($d \geq 2$, non-identity
restriction maps) do not admit the same locality. Experiments on
Barabasi-Albert graphs with up to $5 \times 10^6$ vertices and
$1.5 \times 10^7$ streaming edits show 35 $\mu$s median lazy
per-edit update latency (excluding flush). Exact synchronization
costs are reported separately in \Cref{sec:experiments}.
\end{abstract}

\section{Introduction}\label{sec:intro}

Cellular sheaves are a mathematical framework that assigns vector
spaces (stalks) to the cells of a cell complex and linear maps
(restriction maps) to the incidence relations between cells. The
cohomology groups of a cellular sheaf, computed via the kernel and
image of coboundary operators on cochain spaces, encode global
consistency obstructions: nonzero first cohomology $H^1(X; \FF)$
detects the presence of \emph{structural contradictions} in the data
assigned to the sheaf.

This detection capability has found applications in sensor network
coverage \cite{ghrist2008barcodes, robinson2014topological}, opinion
dynamics \cite{hansen2020opinion}, distributed optimization
\cite{barbero2022sheaf}, neural network architectures
\cite{bodnar2022neural, hansen2020sheafnn}, and knowledge graph
consistency verification \cite{curry2014sheaves}. In all these
settings, the underlying complex is not static: edges arrive, vertices
are inserted, and the data on the sheaf (restriction maps encoding
relationships between entities) evolves over time.

Yet the standard computation of $H^1$ requires forming the coboundary
matrix $\delta^0: C^0(X;\FF) \to C^1(X;\FF)$ and computing its rank
via singular value decomposition (SVD), Gaussian elimination, or
equivalent matrix factorization. For a complex with $n$ cells (vertices
plus edges), the coboundary matrix has $O(n)$ rows and columns (each
of dimension equal to the stalk dimension $d$), and its factorization
costs $O(n^3 d^3)$. When the complex evolves with a stream of $m$
edits, the naive approach of full recomputation after each edit incurs
$O(m n^3 d^3)$ total cost, which is prohibitive for large-scale dynamic
applications.

\subsection{Contributions}

We make the following contributions:

\begin{enumerate}[leftmargin=2em]
\item \textbf{Locality Lemma (\Cref{lem:locality}).}
We prove that for a partitioned cellular complex with bounded cell
size $v_{\max}$, bounded stalk dimension $d$, and bounded nerve
degree $D$, a single edit changes the local coboundary data of at
most 2 cover elements. This formalizes the intuition that ``local
edits have local data effects'' under bounded geometry; recovery of
the global $\dim H^1$ from the updated local data requires
Mayer-Vietoris assembly.

\item \textbf{Incremental Maintenance Theorem (\Cref{thm:incremental}).}
We show that each incremental edit can be processed in
$O(v_{\max}^3 \cdot d^3)$ time, which is $O(1)$ with respect to the
total complex size $n$ when $v_{\max}$ and $d$ are treated as constants
(as they are in practice). This replaces $O(n^3)$ per-edit
recomputation with $O(1)$-in-$n$ lazy edit ingestion, with exact
global assembly deferred to synchronization (\Cref{sec:assembly}).
In dynamic-algorithm terms this is $O(1)$-in-$n$ \emph{update} time;
the \emph{query} time (cost to report $\dim H^1$ on demand) is set by
the assembly path and is treated separately in
\Cref{rem:update-query}.

\item \textbf{Zero-Drift Theorem (\Cref{thm:drift}).}
We prove that the incrementally maintained $H^1$ agrees with
the batch-assembled $H^1$ of the partitioned sheaf model at every
synchronization point ($\textsc{Flush}$), yielding zero measured
drift. Between flushes, the dirty set may contain stale cells;
exactness is restored upon synchronization.

\item \textbf{Streaming Construction (\Cref{thm:streaming}).}
We describe an algorithm that constructs the cellular decomposition
from a raw edge stream, creating cells on demand and splitting them
on overflow, with $O(|E|)$ total time and amortized $O(1)$ per edge.
The amortized split cost is established via a credit
argument (\Cref{thm:split}).

\item \textbf{Adversarial Barrier Argument (\Cref{thm:lower}).}
We discuss why the partition structure appears necessary: for
non-trivial sheaves (stalk dimension $d \geq 2$ with non-identity
restriction maps), we give an adversarial argument in the algebraic
RAM model suggesting that any algorithm maintaining exact $H^1$ on an
unpartitioned complex must perform $\Omega(n)$ work per edit.
Combined with the upper bound, this suggests a characterization:
for non-trivial sheaves, $O(1)$-in-$n$ maintenance requires
bounded local geometry.

\item \textbf{Empirical Validation (\Cref{sec:experiments}).}
We report benchmarks on complexes with up to $V = 5 \times 10^6$
vertices and $E = 1.5 \times 10^7$ edges, demonstrating 35
$\mu$s median \emph{lazy} per-edit cost (excluding flush) that is
\emph{independent of $n$}, with zero measured drift at
synchronization points, verified by full independent recomputation
through $V = 5 \times 10^6$. Exact synchronization costs (eigensolve +
assembly) are reported separately.

\item \textbf{Contradiction Localization (\Cref{sec:localization}).}
We show that a single structural contradiction injected into an
otherwise consistent sheaf at $V = 5 \times 10^6$ is detected by
recomputing one cell out of 25{,}473 (0.0039\% of the complex), with
the global section count collapsing from 8 to 1 and no global
recomputation, producing a deterministic signed receipt. This is the
locality property of \Cref{lem:locality} applied to detection rather
than maintenance.
\end{enumerate}

\subsection{Relation to Prior Work}

The literature on incremental algebraic topology includes work on
persistent homology \cite{cohen2006vines, dey2014computing,
edelsbrunner2010computational}, where the filtration structure permits
efficient updates to Betti numbers as simplices are added. However,
these algorithms exploit the total order of the filtration and are
specific to simplicial homology; they do not directly apply to sheaf
cohomology, where the additional structure of stalks and restriction
maps introduces a richer (and more expensive) algebraic framework.

Dynamic graph algorithms \cite{eppstein1997sparsification, italiano2006fully}
maintain various properties (connectivity, shortest paths, spanning
trees) under edge insertions and deletions in sublinear amortized
time. Our work can be seen as extending this paradigm to a
\emph{cohomological} invariant on a \emph{sheaf-decorated} graph.

Recent work on sheaf neural networks \cite{bodnar2022neural,
hansen2020sheafnn} has brought sheaf-theoretic methods into the machine
learning community, but these approaches treat the sheaf structure
as a fixed architectural choice rather than a dynamically maintained
invariant. The Hodge Laplacian computations in these networks are
$O(n^2)$ to $O(n^3)$ and are performed once at initialization; our
contribution shows how to maintain the cohomological information
incrementally as the underlying data changes.

Incremental matrix factorization and low-rank updates
\cite{brand2006fast, bunch1978rank} provide algebraic machinery for
updating SVD and eigendecompositions under rank-one perturbations.
While these tools are related, they operate on the global coboundary
matrix and achieve $O(n^2)$ per update at best. Our approach avoids
the global matrix entirely by exploiting the partition structure to
localize the computation.

\section{Preliminaries}\label{sec:prelim}

\subsection{Cell Complexes and Cellular Sheaves}

\begin{definition}[Cell Complex]
A \emph{cell complex} $X$ is a finite collection of cells of
various dimensions. We restrict attention to 1-dimensional cell
complexes (graphs): $X$ consists of a finite set of 0-cells
(vertices) $X_0$ and a finite set of 1-cells (edges) $X_1$, with
an incidence relation $\sigma \leq \tau$ when a vertex $\sigma$ is
a face of an edge $\tau$.
\end{definition}

\begin{definition}[Cellular Sheaf]
A \emph{cellular sheaf} $\FF$ on a cell complex $X$ assigns:
\begin{enumerate}[leftmargin=2em]
  \item To each cell $\sigma \in X$, a finite-dimensional real
        vector space $\FF(\sigma)$ called the \emph{stalk} at
        $\sigma$, with $\dim \FF(\sigma) = d_\sigma$.
  \item To each incidence relation $\sigma \leq \tau$, a linear
        map $\FF_{\sigma \leq \tau}: \FF(\sigma) \to \FF(\tau)$
        called a \emph{restriction map}.
\end{enumerate}
When all vertex stalks have dimension $d$ and all edge stalks have
dimension $d$, we say the sheaf has \emph{uniform stalk dimension}
$d$.
\end{definition}

\begin{definition}[Cochain Spaces and Coboundary Operator]
The \emph{$k$-cochain space} is
\[
  C^k(X; \FF) = \bigoplus_{\sigma \in X_k} \FF(\sigma).
\]
For a 1-complex, the \emph{coboundary operator}
$\delta^0: C^0(X;\FF) \to C^1(X;\FF)$ is defined on a 0-cochain
$x = (x_v)_{v \in X_0}$ by
\[
  (\delta^0 x)_e = \FF_{v_+ \leq e}(x_{v_+}) - \FF_{v_- \leq e}(x_{v_-}),
\]
where $v_+, v_-$ are the two vertices incident to edge $e$
(with a fixed but arbitrary orientation convention).
\end{definition}

\begin{definition}[Sheaf Cohomology]
The \emph{first sheaf cohomology group} is
\[
  H^1(X; \FF) = C^1(X;\FF) \,/\, \im(\delta^0).
\]
Its dimension $\dim H^1(X;\FF)$ counts the number of independent
global consistency obstructions in the sheaf data. For a
1-dimensional complex there are no 2-cells, so the next coboundary
operator $\delta^1: C^1 \to C^2$ vanishes identically; consequently
$H^1 = \ker \delta^1 / \im \delta^0 = C^1 / \im \delta^0 =
\operatorname{coker} \delta^0$, and $\dim H^1$ reduces to the single
rank computation $\dim C^1 - \rk(\delta^0)$. This reduction is
specific to the 1-dimensional setting; its consequences for
higher-dimensional complexes are discussed in
\Cref{sec:discussion}.
\end{definition}

\begin{definition}[Sheaf Laplacian]
The \emph{sheaf Laplacian} is
$L_\FF = (\delta^0)^\top \delta^0 : C^0(X;\FF) \to C^0(X;\FF)$.
This is the \emph{0-Laplacian} $L_0$. The kernel of $L_\FF$ is
the space of \emph{global sections} $H^0(X;\FF)$, and since
$\rk(L_\FF) = \rk(\delta^0)$ (the nonzero eigenvalues of $L_\FF$
correspond to the nonzero singular values of $\delta^0$), the
Hodge theorem for cellular sheaves gives
$\dim H^1(X;\FF) = \dim C^1(X;\FF) - \rk(\delta^0)$.
Thus $\dim H^1$ is obtained from the eigendecomposition of $L_0$
without explicitly forming or factoring the 1-Laplacian
$L_1 = \delta^0 (\delta^0)^\top$.
The second smallest eigenvalue $\lambda_2(L_\FF)$, the
\emph{sheaf spectral gap}, quantifies the degree of
near-inconsistency in the sheaf data. While $\lambda_2$
could in principle be maintained incrementally via similar
locality arguments, this paper focuses exclusively on the
topological invariant $\dim H^1$.
\end{definition}

\subsection{Classical Complexity}

\begin{proposition}[Classical Complexity]\label{prop:classical}
For a cellular sheaf with $|X_0| = V$ vertices, $|X_1| = E$ edges,
and uniform stalk dimension $d$, the coboundary matrix
$\delta^0 \in \R^{Ed \times Vd}$ has dimensions
$O(n d) \times O(n d)$ where $n = V + E$. Computing
$\dim H^1 = Ed - \rk(\delta^0)$ via SVD requires $O(n^3 d^3)$
time, or $O(n^3)$ when $d$ is a constant.
\end{proposition}

\begin{proof}
The SVD of an $m \times p$ matrix costs $O(\min(m^2 p, m p^2))$.
Here $m = Ed$ and $p = Vd$, giving a general cost of
$O(\min(E^2 V,\, E V^2)\, d^3)$ for arbitrary graphs (the rectangular
case, relevant when $E \gg V$, i.e.\ dense complexes). For graphs
with $E = O(V)$ under bounded degree, $m, p = O(Vd) = O(nd)$, so this
is $O(n^3 d^3)$.
\end{proof}

\subsection{Partitioned Cell Complexes}

\begin{definition}[Cellular Decomposition]\label{def:decomp}
A \emph{cellular decomposition} of a graph $G = (V, E)$ is a
partition $\{V_1, V_2, \ldots, V_k\}$ of $V$ into disjoint subsets
(cells) such that each cell $V_i$ induces a connected subgraph
$G_i = G[V_i]$, together with:
\begin{enumerate}[leftmargin=2em]
  \item A set of \emph{boundary vertices}: each cross-cell edge
        $(u, v)$ with $u \in V_i$ and $v \in V_j$ is assigned to a
        single \emph{host} cell (in the streaming construction,
        whichever of the two endpoints' cells has more spare
        capacity). The far endpoint is duplicated into the host
        cell as a boundary vertex, and the edge is realized inside
        the host. If $V_i$ is the host, a formal copy $v'$ of $v$ is
        added to $V_i$, carrying the same stalk $\FF(v)$, and the
        edge is realized as $(u, v')$ within $V_i$. The non-host
        cell $V_j$ is not modified internally; its own copy of $v$
        is recorded as a boundary vertex so that the shared
        interface is visible from both sides.
  \item A \emph{nerve complex} $\NN$: the abstract simplicial
        complex whose 0-simplices are the cells $\{V_i\}$ and
        whose 1-simplices connect pairs of cells sharing at least
        one boundary vertex. A $p$-simplex $\sigma$ in $\NN$
        corresponds to a $(p+1)$-fold intersection among cells.
  \item \emph{Boundary restriction maps}: for each boundary vertex
        $v$ shared between cells $V_i$ and $V_j$, a linear map
        $\rho_{ij}^v: \FF(v)|_{V_i} \to \FF(v)|_{V_j}$
        encoding the compatibility condition between the two
        cells' stalks at $v$. These maps inherit the Purity Gate
        bound (\Cref{def:purity}).
\end{enumerate}
The decomposition induces a \emph{local coboundary matrix}
$\delta^0_i$ for each cell $V_i$, constructed from the edges and
restriction maps internal to that cell (including edges to
boundary vertices). The local sheaf Laplacian is
$L_i = (\delta^0_i)^\top \delta^0_i$.
\end{definition}

\begin{remark}[Equivalence of the duplicated-boundary model]
\label{rem:equiv}
The boundary-duplication construction produces a
\emph{sheaf-theoretic covering} of $X$: the cells $\{V_i\}$ with
their boundary vertices form an open cover of $X$ in the
Alexandrov topology on the face poset, and the boundary
restriction maps encode the cocycle condition on overlaps.
The Mayer-Vietoris sequence for this cover
\cite{curry2014sheaves, ghrist2014elementary} recovers the
cohomology of the original (non-duplicated) complex $X$.
Concretely, the duplicated-boundary complex $\tilde{X}$
(with formal copies of boundary vertices) is not used as an
independent replacement for $X$; it is an implementation device
that represents the cover and its overlap data. In general
$\dim H^1(\tilde{X}; \tilde{\FF}) \neq \dim H^1(X; \FF)$, since
duplicating boundary vertices changes the topology of the
underlying complex (cf.\ \Cref{rem:c4}). The Mayer-Vietoris
assembly recovers the original $\dim H^1(X; \FF)$ from the local
cell data and boundary maps precisely by imposing the
compatibility (cocycle) constraints on overlaps; it is the
assembly, not $\tilde{X}$ itself, that carries the equivalence.
\end{remark}

\begin{definition}[Bounded Local Geometry]\label{def:blg}
A cellular decomposition has \emph{bounded local geometry} with
parameters $(v_{\max}, d, D)$ if:
\begin{enumerate}[leftmargin=2em]
  \item Each cell contains at most $v_{\max}$ vertices (including
        boundary duplicates).
  \item All stalks have dimension at most $d$.
  \item Each cell is adjacent to at most $D$ other cells in the
        nerve complex $\NN$.
\end{enumerate}
All three parameters are independent of $n = |V|$.
\end{definition}

\subsection{Computational Model}\label{sec:model}

We work in the \emph{algebraic RAM model}: the machine has
random-access memory over $\R$ (or a finite-precision
approximation thereof), each arithmetic operation ($+$, $-$,
$\times$, $\div$) on a pair of real numbers costs $O(1)$, and
each comparison costs $O(1)$. Hash table operations (insert,
lookup, delete) cost $O(1)$ amortized. Memory words hold real
numbers or $O(\log n)$-bit integers.

Under this model, multiplying two $d \times d$ matrices costs
$O(d^3)$, computing the eigendecomposition of a symmetric
$m \times m$ matrix costs $O(m^3)$, and a rank computation on an
$m \times p$ matrix costs $O(\min(m^2 p, m p^2))$.

All complexity bounds in this paper are stated in terms of $n$
(total complex size), $v_{\max}$ (cell size bound), $d$ (stalk
dimension), and $D$ (nerve degree bound). When we write ``$O(1)$
in $n$,'' we mean that the cost is bounded by a function of
$v_{\max}$, $d$, and $D$ alone, with no dependence on $n$.

\section{Main Results}\label{sec:main}

We now state and prove the main results. Throughout, we assume a
cellular decomposition with bounded local geometry $(v_{\max}, d, D)$
as in \Cref{def:blg}.

The algorithm maintains $\dim H^1(X; \FF)$ as defined on the
original (non-duplicated) complex, not a local proxy or
approximation. At each synchronization point (\textsc{Flush}),
the global dimension is recovered from local cell data via
Mayer-Vietoris assembly over the nerve complex
(\Cref{thm:gluing} in \Cref{sec:assembly}). The equivalence
between the duplicated-boundary model used by the algorithm and
the original sheaf cohomology is established in
\Cref{rem:equiv}. Between flushes, local cell data may be stale,
but the dirty set tracks all cells requiring recomputation; the
staleness is fully resolved at the next flush.

\subsection{The Locality Lemma}

\begin{lemma}[Locality of Cohomological Impact]\label{lem:locality}
Let $X$ be a cell complex with cellular decomposition
$\{V_1, \ldots, V_k\}$ having bounded local geometry
$(v_{\max}, d, D)$. Let $\FF$ be a cellular sheaf on $X$, and let
$\FF'$ be a sheaf that differs from $\FF$ by a single edit of one
of the following types:
\begin{enumerate}[leftmargin=2em]
  \item[(i)] \textbf{Intra-cell edge insertion}: adding an edge
    between two vertices within the same cell $V_i$.
  \item[(ii)] \textbf{Cross-cell edge insertion}: adding an edge
    between vertices in distinct cells $V_i, V_j$, hosted by one of
    them, with the far endpoint duplicated into the host cell as a
    boundary vertex.
  \item[(iii)] \textbf{Vertex insertion}: adding a new vertex to
    an existing cell $V_i$ with one incident edge.
  \item[(iv)] \textbf{Restriction map update}: replacing the
    restriction maps on an existing edge within cell $V_i$.
\end{enumerate}
Then the edit changes the local coboundary data of exactly one
cover element, the host cell, in all four cases. In case (ii) a
second cell is additionally affected, but not through its local
coboundary: the new boundary restriction map (and the new
boundary-vertex status of the shared vertex) couples the host and
the non-host cell, so the non-host cell's contribution to the
global cohomology can change even though its own coboundary
$\delta^0_j$ is unchanged. Consequently at most two cells require
recomputation before the next global assembly step.

Specifically, let $S \subseteq \{1, \ldots, k\}$ denote the set of
cell indices whose local coboundary matrix $\delta^0$ changes, and
let $R \supseteq S$ denote the set of cells that must be recomputed
before assembly. Then:
\begin{enumerate}[leftmargin=2em]
  \item $|S| = 1$ in all cases (i)--(iv): only the host cell's local
    coboundary changes.
  \item $|R| = 1$ for cases (i), (iii), and (iv); and $|R| = 2$ for
    case (ii), where $R = S \cup \{j\}$ adds the non-host cell
    coupled through the new boundary restriction map.
\end{enumerate}
For all cells $V_\ell$ with $\ell \notin S$, the local coboundary
$\delta^0_\ell$ is unchanged. Recovery of the global
$\dim H^1(X; \FF)$ from the updated local data requires
Mayer-Vietoris assembly over the nerve complex (\Cref{thm:gluing}),
and it is through that assembly step that the case-(ii) boundary map
affects the non-host cell.
\end{lemma}

\begin{proof}
We consider each case separately.

\textbf{Case (i): Intra-cell edge insertion.}
Let $e = (u, v)$ be a new edge with $u, v \in V_i$. The global
coboundary matrix $\delta^0$ is a block matrix indexed by edges
(rows) and vertices (columns). Adding edge $e$ appends one block
row to $\delta^0$, with nonzero entries only in the columns
corresponding to $u$ and $v$, both of which belong to cell $V_i$.
No other cell's block of $\delta^0$ is affected. Hence the local
coboundary $\delta^0_j$ for $j \neq i$ is unchanged.

\textbf{Case (ii): Cross-cell edge insertion.}
Let $e = (u, v)$ with $u \in V_i$, $v \in V_j$, $i \neq j$, and let
$V_i$ be the host cell. The algorithm duplicates $v$ into $V_i$ as a
boundary vertex $v'$, adds the edge $(u, v')$ within $V_i$, and
installs a boundary restriction map coupling $V_i$ and $V_j$ at the
shared vertex. Only $\delta^0_i$ changes (one new edge and one new
boundary vertex in cell $i$); the non-host cell $V_j$ receives no new
interior edge or vertex, so its local coboundary $\delta^0_j$ is
unchanged, giving $|S| = 1$. However, the boundary restriction map
is inter-cell data: it enters the Mayer-Vietoris assembly
differential on nerve edge $(i, j)$ (\Cref{def:assembly-complex})
and is consulted when the global computation resolves the cohomology
contribution attributed to $V_j$. That contribution can therefore
change with no edit to the local sheaf of $V_j$, so $V_j$ must also
be recomputed: $R = \{i, j\}$, $|R| = 2$. No cell $V_\ell$ with
$\ell \notin \{i, j\}$ is affected.

\textbf{Case (iii): Vertex insertion.}
Adding a new vertex $w$ to cell $V_i$ with one incident edge
$(w, u)$ for $u \in V_i$ appends one block row and one block
column to $\delta^0_i$. No other cell is affected. $|S| = 1$.

\textbf{Case (iv): Restriction map update.}
Replacing $\FF_{u \leq e}$ and $\FF_{v \leq e}$ on an existing
edge $e = (u, v)$ within cell $V_i$ modifies the entries of
$\delta^0_i$ in the block row for $e$. No other cell's coboundary
matrix is affected. $|S| = 1$.
\end{proof}

\begin{remark}[Pathological topologies]
The bounded local geometry assumption is essential. In an
unpartitioned complex (equivalently, a single cell containing
all $n$ vertices), every edit potentially affects the entire
coboundary matrix, and no locality can be exploited. Our algorithm
gains its advantage precisely from the cellular decomposition
structure, which localizes each edit's impact to $O(v_{\max})$
vertices rather than $O(n)$.
\end{remark}

\subsection{The Incremental Maintenance Theorem}

\begin{theorem}[Incremental Maintenance]\label{thm:incremental}
Let $X$ be a cell complex with cellular decomposition having bounded
local geometry $(v_{\max}, d, D)$, and consider a single edit of type
(i)--(iv) from \Cref{lem:locality}. The algorithm
$\textsc{IncrementalUpdate}$ satisfies:
\begin{enumerate}[leftmargin=2em]
  \item \emph{(Local update.)} The affected local coboundary data can
    be updated, and the rank of at most 2 affected cells recomputed,
    in $O(v_{\max}^3 \cdot d^3)$ time.
  \item \emph{(Certified eager assembly.)} If the algorithm
    additionally maintains an incremental rank certificate for the
    nerve-level Mayer-Vietoris differential, then the corresponding
    global update of $\dim H^1(X; \FF)$ costs $O(D^3 \cdot d^3)$,
    for a total exact eager update of
    $O(v_{\max}^3 d^3 + D^3 d^3)$ per edit.
  \item \emph{(Implemented lazy path.)} The implementation evaluated
    in \Cref{sec:experiments} uses the lazy edit path of
    \Cref{thm:lazy} (cost $O(d^2)$ per edit, with exact global
    $\dim H^1$ restored at \textsc{Flush}), and performs global
    assembly as a full nerve traversal at each synchronization point
    rather than via the certificate of case~2.
\end{enumerate}
Since $v_{\max}$, $d$, and $D$ are constants independent of
$n = |V|$, both the local update and the certified eager update are
$O(1)$ per edit with respect to $n$; the lazy ingestion cost is
likewise $O(1)$ in $n$, while the full-traversal assembly at
\textsc{Flush} is $O(n)$ per synchronization point
(\Cref{prop:assembly}).
\end{theorem}

\begin{proof}
By \Cref{lem:locality}, each edit affects at most 2 cells.
For each affected cell $V_i$, the local coboundary matrix
$\delta^0_i \in \R^{E_i d \times V_i d}$ has $O(E_i d)$ rows and
$O(V_i d)$ columns, with $V_i \leq v_{\max}$ and
$E_i \leq \binom{v_{\max}}{2} = O(v_{\max}^2)$; constructing it
from the sparse local incidence data costs time polynomial in
$v_{\max}$ and $d$, hence $O(1)$ in $n$. The rank of $\delta^0_i$
is obtained from the local $0$-Laplacian
$L_i = (\delta^0_i)^\top \delta^0_i \in \R^{V_i d \times V_i d}$,
which is $O(v_{\max} d) \times O(v_{\max} d)$; its dense
eigendecomposition costs $O((v_{\max} d)^3) = O(v_{\max}^3 d^3)$.
This establishes case~1: the local update over at most 2 cells
costs $2 \cdot O(v_{\max}^3 d^3) = O(v_{\max}^3 d^3)$.

For case~2, recovery of the global $\dim H^1(X; \FF)$ proceeds via
Mayer-Vietoris assembly over the nerve complex $\NN$
(\Cref{sec:assembly}). An edit affecting at most 2 cells changes
at most $O(D)$ rows of the \v{C}ech differential $\partial_0$.
With a maintained rank factorization of the \v{C}ech differentials,
the affected Schur complement has dimension $O(D \cdot d)$ and the
rank update costs $O(D^3 d^3)$, giving a total exact eager update of
$O(v_{\max}^3 d^3 + D^3 d^3)$, independent of $n$.

Case~3 is immediate: the lazy edit cost is $O(d^2)$ by
\Cref{thm:lazy}, and the full-traversal assembly invoked at
\textsc{Flush} is analyzed in \Cref{prop:assembly}.
\end{proof}

\begin{corollary}[Total Streaming Cost]
\label{cor:total}
Processing a stream of $m$ edits on a complex with $n$ total
vertices costs $O(m \cdot (v_{\max}^3 d^3 + D^3 d^3))$ total time
in the certified eager model (case~2 of
\Cref{thm:incremental}), compared to $O(m \cdot n^3 d^3)$ for
full recomputation after each edit. The speedup factor in the
certified eager model is
\[
  \frac{n^3 d^3}{v_{\max}^3 d^3 + D^3 d^3}
  \approx \frac{n^3}{v_{\max}^3}.
\]
For a knowledge graph with $n = 10^6$ and
$v_{\max} = 500$ (the empirical $D$ satisfies $D < v_{\max}$,
so the $v_{\max}^3$ term dominates the denominator and the $D$
term is negligible), this factor is approximately
$8 \times 10^9$. This speedup applies to the certified eager
path, which is not yet implemented; the implemented lazy path
achieves $O(1)$-in-$n$ \emph{update} time but incurs $O(n)$
\emph{query} time at each flush (\Cref{rem:update-query}).
\end{corollary}

\subsection{Lazy Evaluation and Amortized Analysis}

The algorithm supports two execution modes:

\begin{definition}[Eager and Lazy Modes]
In \emph{eager mode}, the local eigensolve and global assembly are
performed immediately after each edit. In \emph{lazy mode}, the edit
marks the host cell as \emph{dirty} and defers both the eigensolve
and assembly to a batch $\textsc{Flush}$ operation. Only the dirty
set membership test (hash set insertion, $O(1)$) is performed per
edit.
\end{definition}

\begin{theorem}[Lazy Amortized Cost]\label{thm:lazy}
In lazy mode, the per-edit cost is $O(d^2)$ (constant time for
constant $d$, independent of all other parameters). The
$\textsc{Flush}$ operation processes all $k_d$ dirty cells in
$O(k_d \cdot v_{\max}^3 d^3 + k_d \cdot D^3 \cdot d^3)$ time,
where the first term covers local eigensolves and the second
covers global assembly. When $\textsc{Flush}$ is called after
every $B$ edits, the amortized cost per edit is
\[
  O(d^2) + \frac{O(k_d \cdot (v_{\max}^3 d^3 + D^3 d^3))}{B}.
\]
When $B = \Omega(k_d)$, this simplifies to
$O(v_{\max}^3 d^3 + D^3 d^3)$ amortized per edit, which is $O(1)$
in $n$.
\end{theorem}

\begin{proof}
Each lazy edit performs: (a) a hash lookup to determine case
classification, $O(1)$; (b) graph mutation (add vertex/edge to a
local adjacency list), $O(1)$; (c) restriction map initialization
(copy or reference from a pre-computed pool), $O(d^2)$; and (d)
dirty set insertion, $O(1)$. Total: $O(d^2) = O(1)$ for constant
$d$.

At flush time, each dirty cell undergoes one eigensolve of its
local Laplacian, costing $O(v_{\max}^3 d^3)$. After all dirty
cells are recomputed, the global assembly step recovers
$\dim H^1(X; \FF)$ from the local data and boundary maps via
the Mayer-Vietoris sequence over the nerve complex. With a
maintained rank factorization of the \v{C}ech differentials,
the assembly update costs $O(k_d \cdot D^3 \cdot d^3)$; without
such a certificate, a full nerve traversal is required.
In the implementation evaluated here, assembly is performed as
a full nerve pass at each synchronization point.
\end{proof}

\begin{remark}[Update versus query complexity]\label{rem:update-query}
Following the standard convention in dynamic-algorithm analysis, we
distinguish \emph{update} time (the cost to ingest one edit) from
\emph{query} time (the cost to report $\dim H^1$ on demand). The lazy
per-edit \emph{update} cost is $O(d^2)$, i.e.\ $O(1)$ in $n$. A
\emph{query} forces a \textsc{Flush}, whose cost depends on the
assembly path. The amortized statement of \Cref{thm:lazy} uses the
certificate-path flush cost
$O(k_d \cdot (v_{\max}^3 d^3 + D^3 d^3))$, for which $B = \Omega(k_d)$
suffices to keep the amortized update $O(1)$ in $n$ (since
$k_d \leq B$). In the implementation evaluated here, however, assembly
is a full nerve traversal costing $O(|\NN| \cdot D^2 d^3) = O(n)$ per
flush (\Cref{prop:assembly}); in that regime the $O(n)$ assembly term
amortizes to $O(1)$ per edit only when $B = \Omega(n)$, and a query
issued after every edit costs $O(n)$ per query. In short: with the
full-traversal assembly, the \emph{update} time is $O(1)$ in $n$ while
the \emph{query} time is $O(n)$; reducing query time to $O(1)$ in $n$
requires the incremental \v{C}ech certificate of
\Cref{prop:assembly}. The benchmarks in \Cref{sec:experiments} report
\emph{update} latency (lazy per-edit cost, excluding flush); flush and
query costs are reported separately.
\end{remark}

\subsection{The Zero-Drift Theorem}

We first state the invariant that the algorithm maintains, then
prove that it implies exactness.

\begin{definition}[Dirty-Set Invariant]\label{def:invariant}
Let $\mathcal{D} \subseteq \{1, \ldots, k\}$ denote the dirty set.
The algorithm maintains the following invariant after every edit
and after every flush:
\begin{enumerate}[leftmargin=2em]
  \item[(I1)] For every cell index $i \notin \mathcal{D}$, the
    cached local coboundary matrix $\hat{\delta}^0_i$ equals the
    true coboundary matrix $\delta^0_i$ constructed from the
    current graph and sheaf data of cell $V_i$ ($\hat{\delta}^0_i
    = \delta^0_i$), and no boundary restriction map incident to
    $V_i$ has changed since its cached cohomology contribution was
    last computed.
  \item[(I2)] For every cell index $i \in \mathcal{D}$, the cached
    $\hat{\delta}^0_i$ or the cached cohomology contribution of
    $V_i$ may be stale (it reflects the state at the time of the
    last flush or initialization, not the current state).
  \item[(I3)] The dirty set $\mathcal{D}$ contains every cell whose
    cached data may no longer reflect the current sheaf: every cell
    $i$ with a stale local coboundary
    ($\hat{\delta}^0_i \neq \delta^0_i$), and, for a case-(ii) edit,
    the non-host cell, whose shared vertex has just become a
    boundary vertex with a new boundary restriction map coupling it
    to the host (\Cref{lem:locality}). Equivalently, no cell outside
    $\mathcal{D}$ has either a stale local coboundary or a changed
    incident boundary coupling.
  \end{enumerate}
\end{definition}

\begin{lemma}[Invariant Preservation]\label{lem:invariant}
Each edit of type (i)--(iv) from \Cref{lem:locality} preserves
the dirty-set invariant (\Cref{def:invariant}).
\end{lemma}

\begin{proof}
We verify each case:

\textbf{Case (i), (iii), (iv):} The edit modifies $\delta^0_i$ for
exactly one cell $V_i$ (by \Cref{lem:locality}) and changes no
boundary restriction map. The algorithm inserts $i$ into
$\mathcal{D}$. For all $j \neq i$, the edit does not change
$\delta^0_j$ or any boundary map incident to $V_j$, so if
$j \notin \mathcal{D}$ before the edit, condition (I1) continues to
hold for $j$. Condition (I3) is maintained because $i$ is now in
$\mathcal{D}$.

\textbf{Case (ii):} The edit modifies $\delta^0_i$ for the host
cell $V_i$, marks the shared vertex as a boundary vertex in both
cells, and installs a boundary restriction map coupling $V_i$ and
the non-host cell $V_j$. The algorithm inserts both $i$ and $j$
into $\mathcal{D}$ ($R = \{i, j\}$ in \Cref{lem:locality}). The
non-host coboundary $\delta^0_j$ is unchanged, but its incident
boundary coupling changed, so (I3) requires $j \in \mathcal{D}$,
which holds. For all $\ell \notin \{i, j\}$, the edit changes
neither $\delta^0_\ell$ nor any boundary map incident to $V_\ell$,
so (I1) continues to hold for $\ell$. Condition (I3) is maintained
because both affected cells are now in $\mathcal{D}$.

In all cases, (I2) is satisfied trivially for cells in
$\mathcal{D}$, since stale caches are permitted for dirty cells.
\end{proof}

\begin{theorem}[Zero Drift]\label{thm:drift}
Let $X_0, X_1, \ldots, X_m$ be a sequence of cell complexes
obtained by applying edits $e_1, \ldots, e_m$ to an initial
complex $X_0$. Let $h^1_{\text{inc}}(t)$ denote the $H^1$
dimension maintained by the incremental algorithm after
$\textsc{Flush}$ at step $t$, and let $h^1_{\text{batch}}(t)$
denote the $H^1$ dimension obtained by full recomputation of
$H^1(X_t; \FF_t)$ from scratch. Then for all $t$:
\[
  h^1_{\text{inc}}(t) = h^1_{\text{batch}}(t).
\]
The drift $\Delta(t) = h^1_{\text{inc}}(t) - h^1_{\text{batch}}(t)$
is exactly zero for all $t$.
\end{theorem}

\begin{proof}
We proceed by induction on the number of flushes.

\textbf{Base case ($t = 0$):} The initial complex is built by batch
computation, so $h^1_{\text{inc}}(0) = h^1_{\text{batch}}(0)$
trivially. The dirty set is empty, and (I1) holds for all cells.

\textbf{Inductive step:} Assume $h^1_{\text{inc}}(t-1) =
h^1_{\text{batch}}(t-1)$. Between flush $t-1$ and flush $t$,
a set of edits modifies certain cells. Let $\mathcal{D}$ be the
dirty set at flush $t$. By \Cref{lem:invariant}, the dirty-set
invariant holds throughout the edit sequence. The flush operation
performs two phases:

\emph{Phase 1 (Local recomputation):} For each $V_i$ with
$i \in \mathcal{D}$, the local coboundary matrix $\delta^0_i$ is
reconstructed from the current graph and sheaf data, and
$\rk(\delta^0_i)$ is recomputed via exact eigendecomposition of
the local 0-Laplacian
$L_i = (\delta^0_i)^\top \delta^0_i$ (counting nonzero
eigenvalues). The local $\dim H^1_i = \dim C^1_i - \rk(\delta^0_i)$
is then determined. After this phase,
$\hat{\delta}^0_i = \delta^0_i$ for all $i \in \mathcal{D}$.
Combined with invariant (I1) for $i \notin \mathcal{D}$, we have
$\hat{\delta}^0_i = \delta^0_i$ for all $i \in \{1, \ldots, k\}$:
every cell's cached data is now current.

\emph{Phase 2 (Global assembly):} The global $\dim H^1(X; \FF)$
is recovered from the local cell data and inter-cell boundary
restriction maps via the Mayer-Vietoris exact sequence over the
nerve complex $\NN$
\cite{curry2014sheaves, ghrist2014elementary}. This step is
necessary because $\dim H^1(X; \FF)$ is a global invariant that
cannot, in general, be obtained by summing local $H^1$ values:
cohomology classes that span multiple cells (cycles in the nerve
whose restriction maps are collectively inconsistent) are
invisible to any single cell. The Mayer-Vietoris sequence recovers
these cross-cell classes from the boundary data. In the
implementation evaluated here, the assembly pass traverses the
full nerve complex at each flush; incremental assembly
certificates that would restrict the traversal to
dirty-cell neighborhoods are discussed as future work in
\Cref{sec:discussion}.

Since Phase 1 restores exact local data for all cells and Phase 2
applies the standard Mayer-Vietoris sequence, which computes the
global $H^1$ from local data and boundary maps, the result
$h^1_{\text{inc}}(t)$ equals the batch-recomputed value
$h^1_{\text{batch}}(t)$.
\end{proof}

\begin{remark}[Exactness model]\label{rem:exact-model}
\Cref{thm:drift} is stated in the algebraic RAM model of
\Cref{sec:model}, where rank is computed exactly. In
finite-precision arithmetic the local eigensolves determine ranks
relative to a numerical tolerance $\epsilon$; ``zero drift'' then
holds relative to that tolerance, and the Purity Gate
(\Cref{def:purity}) bounds the restriction-map norms that govern
conditioning. The finite-precision behavior over very long edit
streams, and mitigations such as periodic algebraic resets, are
discussed in the Limitations (\Cref{sec:discussion}).
\end{remark}

\begin{remark}[The Case B subtlety]
An earlier version of the algorithm marked only the host cell dirty
for cross-cell edge insertions, violating invariant (I3). This
produced a drift of $-2$ on a $V = 21{,}000$ test case (seed 137),
where two $H^1$ classes in the non-host cell were missed. The
fix, marking both cells dirty, costs approximately 10 nanoseconds
of additional overhead (two hash set insertions) and restores the
invariant. This illustrates that the zero-drift property depends on
the completeness of the dirty set: every cell whose coboundary is
affected must be tracked, including cells coupled through boundary
restriction maps.
\end{remark}

\subsection{Streaming Construction}

\begin{theorem}[Streaming Construction]\label{thm:streaming}
Given an edge stream $e_1, e_2, \ldots, e_m$ over a vertex set $V$
(vertices discovered as they appear as edge endpoints), there exists
an algorithm $\textsc{StreamingBuild}$ that constructs a cellular
decomposition with bounded cell size $v_{\max}$ and stalk dimension
$d$ in $O(m)$ total time, creating cells on demand and splitting
them when they exceed $v_{\max}$ vertices. The construction enforces
a vertex multiplicity cap $\mu(v) \leq K$ (with $K = 3$ in the
implementation), which by \Cref{cor:mu-cap-D} guarantees
$D \leq v_{\max} \cdot (K - 1)$ for all input streams. The full
bounded local geometry condition of \Cref{def:blg} therefore holds
a priori, and the incremental maintenance bounds of
\Cref{thm:incremental} apply to arbitrary graph topologies. The
sheaf cohomology $H^1$ is maintainable throughout the construction
via the incremental algorithm.
\end{theorem}

\begin{proof}
The streaming builder processes each edge by classifying it into
one of four cases:

\textbf{Case A (Intra-cell):} Both endpoints are in a common cell.
Delegated to the incremental updater. Cost: $O(d^2)$ in lazy mode.

\textbf{Case B (Cross-cell):} Endpoints are in disjoint cells.
One endpoint is duplicated into the other cell. Cost: $O(d^2)$ in
lazy mode, plus $O(1)$ for nerve registration.

\textbf{Case C (One new vertex):} One endpoint exists, the other
is new. The new vertex is added to the existing endpoint's cell.
Cost: $O(d^2)$ in lazy mode.

\textbf{Case D (Both new):} Neither endpoint exists. A new seed
cell is created containing just these two vertices and one edge.
Cost: $O(d^2)$ for the minimal cell construction.

Each case is $O(d^2) = O(1)$ for constant $d$. After each insertion,
the builder checks if the host cell exceeds $v_{\max} + \tau$
vertices (where $\tau$ is a soft overage tolerance to prevent
cascade splitting). If so, the cell is split via Fiedler spectral
bisection (for cells up to a threshold size) or BFS balanced
partition (for larger cells).

The amortized cost of splits is $O(v_{\max}^2 d^3)$ per insertion
by \Cref{thm:split} below. Total construction cost:
$O(m \cdot (d^2 + v_{\max}^2 d^3)) = O(m)$ for constant
$v_{\max}$ and $d$.
\end{proof}

\begin{theorem}[Amortized Split Cost]\label{thm:split}
The amortized cost of cell splits during streaming construction
is $O(v_{\max}^2 d^3)$ per edge insertion, which is $O(1)$ in $n$.
\end{theorem}

\begin{proof}
Each edge insertion deposits $2 \, c_s / v_{\max}$ credits into
the host cell, where $c_s = O(v_{\max}^3 d^3)$ is the worst-case
cost of a single cell split, comprising the Fiedler bisection of
the graph Laplacian ($O(v_{\max}^3)$) plus sheaf cohomology
eigensolves on the two child cells ($O(v_{\max}^3 d^3)$ each);
the sheaf eigensolve dominates since $d \geq 1$.

\textbf{Per-insertion (no split):} The actual cost is $O(d^2)$.
The credit deposit is $2 \, c_s / v_{\max} = O(v_{\max}^2 d^3)$.
The amortized cost per insertion is therefore
$O(d^2) + O(v_{\max}^2 d^3) = O(v_{\max}^2 d^3)$.

\textbf{Per-split:} A cell is split only when it reaches
$v_{\max} + \tau$ vertices. Since each child of a split has at
most $v_{\max}/2 + \tau/2$ vertices, the cell must have received
at least $v_{\max}/2$ insertions since its creation or last
split. Each insertion deposited $2 \, c_s / v_{\max}$ credits,
so the cell has accumulated at least
\[
  \frac{2 \, c_s}{v_{\max}} \cdot \frac{v_{\max}}{2} = c_s
\]
credits, which pays for the split exactly.

Therefore, the amortized cost per insertion is
$O(v_{\max}^2 d^3)$, which is $O(1)$
in $n$ for constant $v_{\max}$ and $d$.
\end{proof}

\begin{proposition}[Bounded Nerve Degree for Bounded-Degree Graphs]
\label{prop:bounded-D}
Let $G$ be a graph with maximum vertex degree $\Delta$, and let
$\{V_1, \ldots, V_k\}$ be a cellular decomposition with
$|V_i| \leq v_{\max}$ for all $i$. Then the nerve degree of
each cell satisfies
$D(V_i) \leq v_{\max} \cdot \Delta$.
In particular, when both $v_{\max}$ and $\Delta$ are constants
independent of $n$, the full bounded local geometry condition
of \Cref{def:blg} holds with $D = v_{\max} \cdot \Delta$, and
all complexity bounds of \Cref{thm:incremental} apply.
\end{proposition}

\begin{proof}
The nerve degree $D(V_i)$ is the number of distinct cells
adjacent to $V_i$ in the nerve. Two cells are adjacent if and
only if they share at least one boundary vertex. Each boundary
vertex in $V_i$ arises from a cross-cell edge incident to some
vertex $v \in V_i$. Since $v$ has at most $\Delta$ neighbors in
$G$, it contributes at most $\Delta$ cross-cell edges.
Multiple cross-cell edges from the same vertex may connect
$V_i$ to the same neighboring cell (when $v$ has several
neighbors in a single adjacent cell), so $\Delta$ is an upper
bound on the number of \emph{distinct} cells reached from $v$,
not necessarily tight. Summing over all
$|V_i| \leq v_{\max}$ vertices gives at most
$v_{\max} \cdot \Delta$ distinct adjacent cells. In practice
$D \ll v_{\max} \cdot \Delta$ because most edges are intra-cell
and cross-cell edges from a single vertex often land in the
same neighboring cell.
\end{proof}

\begin{corollary}[Degree-Independent Nerve Bound via Multiplicity Cap]
\label{cor:mu-cap-D}
Let $\{V_1, \ldots, V_k\}$ be a cellular decomposition with
$|V_i| \leq v_{\max}$ for all $i$, and suppose the construction
enforces a vertex multiplicity cap $\mu(v) \leq K$ for every
vertex $v$ (i.e., each vertex appears in at most $K$ cells).
Then the nerve degree of each cell satisfies
\[
  D(V_i) \leq v_{\max} \cdot (K - 1),
\]
independently of the maximum vertex degree $\Delta$ of the
underlying graph. In particular, $D = O(1)$ in $n$ for any
graph family, including scale-free graphs with unbounded
$\Delta$. The nerve dimension is at most $K - 1$, so the
assembly complex of \Cref{def:assembly-complex} has at most
$K$ nonzero terms.
\end{corollary}

\begin{proof}
Each vertex $v \in V_i$ appears in at most $\mu(v) \leq K$
cells (including $V_i$ itself), so $v$ connects $V_i$ to at
most $K - 1$ other cells through boundary duplication.
Multiple boundary vertices from $V_i$ may connect to the
same neighboring cell, so the number of \emph{distinct}
adjacent cells is at most $|V_i| \cdot (K - 1) \leq
v_{\max} \cdot (K - 1)$. Since $v_{\max}$ and $K$ are
constants independent of $n$, this bound is $O(1)$ in $n$.

For the nerve dimension: a $(p+1)$-fold intersection among
cells $V_{i_0} \cap \cdots \cap V_{i_p}$ is nonempty only
if some vertex $v$ appears in all $p+1$ cells, requiring
$\mu(v) \geq p+1$. Since $\mu(v) \leq K$, we need
$p + 1 \leq K$, so $p \leq K - 1$. The nerve has no
simplices of dimension $\geq K$; with $K = 3$, the
highest-dimensional simplices are 2-simplices (triangles).
\end{proof}

\begin{remark}[Role of the multiplicity cap]
\label{rem:mu-cap}
The multiplicity cap $\mu(v) \leq K$ is enforced by the
streaming construction: when a cross-cell edge would duplicate
a vertex $v$ into its $(K+1)$-th cell, the duplication direction
is reversed so that the \emph{other} endpoint is duplicated
instead (into a cell where $v$ already resides). This reversal
preserves edge locality while respecting the cap, at the cost
of concentrating boundary duplications on lower-degree vertices.
The batch partition path enforces the same cap during boundary
expansion. Both paths guarantee $\mu_{\max} \leq K = 3$ for
all graph topologies, as confirmed empirically in
\Cref{sec:experiments}.

\Cref{cor:mu-cap-D} supersedes \Cref{prop:bounded-D} for
graph families with unbounded degree: the degree-independent
bound $D \leq v_{\max} \cdot (K - 1)$ applies to
Barabasi--Albert, power-law, and arbitrary graphs without
requiring bounded $\Delta$. \Cref{prop:bounded-D} remains
tighter for bounded-degree families where $\Delta < K - 1$.
\end{remark}

\begin{remark}[Applicability to common graph families]
\label{rem:graph-families}
\Cref{prop:bounded-D} applies directly to graph families with
bounded maximum degree: $d$-regular graphs, lattice graphs,
bounded-degree expanders, and Watts--Strogatz small-world graphs
with fixed rewiring parameter $k$. For Barabasi--Albert graphs
$\mathrm{BA}(n, m)$, the maximum degree grows as
$O(\sqrt{n})$, so \Cref{prop:bounded-D} gives
$D = O(v_{\max} \sqrt{n})$, which is not $O(1)$. However,
\Cref{cor:mu-cap-D} provides a degree-independent bound:
$D \leq v_{\max} \cdot (K - 1) = 1000$ with $K = 3$ and
$v_{\max} = 500$, which is $O(1)$ in $n$ regardless of $\Delta$.

The \emph{empirical} nerve degree depends on the cell count and
partition topology, and on scale-free graphs it can be large.
At $V = 21{,}000$ (203 cells), a few hub-containing cells exhibit
near-complete nerve adjacency: the maximum observed nerve degree
was $D_{\max} = 197$, while the median was $178$.
At $V = 5 \times 10^6$ (62{,}303 cells, 14{,}999{,}994 edges),
the nerve degree reaches $D_{\max} = 992$ with a median of
$441$, approaching but not exceeding the theoretical bound
$v_{\max}(K-1) = 1000$ of \Cref{cor:mu-cap-D}.
This growth is expected: on a Barabasi--Albert graph, hub
vertices participate in many cells, and each cell containing
a hub becomes adjacent to nearly every cell the hub touches.
Critically, this does not affect the per-edit cost: as
\Cref{rem:nerve-pressure} makes precise, the nerve degree
enters only the deferred global-assembly term, while the
per-edit update and the per-cell flush eigensolve are governed
by $v_{\max}$ and the multiplicity cap $\mu_{\max} \leq 3$,
both independent of $D$. The operative bounded objects are
vertex participation ($\mu_{\max}$) and cell size
($v_{\max}$), not global nerve adjacency.
\end{remark}

\begin{remark}[Nerve degree is not the operative bound]
\label{rem:nerve-pressure}
Of the three bounded-local-geometry parameters
$(v_{\max}, d, D)$, the per-edit update cost and the per-cell
flush eigensolve depend only on $v_{\max}$ and $d$, never on the
nerve degree $D$. Three facts establish this. (i) The lazy
per-edit update is $O(d^2)$ (\Cref{thm:lazy}), with no $D$ term.
(ii) A single edit marks at most two cells dirty
(\Cref{lem:locality}); for a Case~B edit these are the host and
the one other cell sharing the duplicated boundary vertex, a
count independent of how many cells the host borders in the
nerve (its degree $D$). (iii) The local recompute of each dirty cell
is $O(v_{\max}^3 d^3)$, again with no $D$ term. The nerve degree
$D$ enters \emph{exclusively} through the global assembly
differential $\partial_0$ (\Cref{prop:assembly}), whose block
rows per cell scale with $D$, and the assembly is the deferred,
synchronization-time step, already $O(n)$ per flush in the
implemented full-traversal path.

Consequently a graph may exhibit high nerve degree (hub cells
adjacent to nearly all others, as observed at $D_{\max} = 992$
at $V = 5 \times 10^6$) while the headline $35\,\mu$s per-edit
latency and the zero-drift guarantee are entirely unaffected.
The invariant the streaming construction actually enforces is
the vertex multiplicity cap $\mu_{\max} \leq K = 3$; the
nerve-degree bound $D \leq v_{\max}(K-1)$ of
\Cref{cor:mu-cap-D} is a \emph{consequence} of that cap, but the
empirical $D$ is not the quantity that governs incremental cost.
The sharper reading of bounded local geometry is therefore: what
is capped is vertex participation ($\mu_{\max} = 3$) and cell
size ($v_{\max} = 500$), not global nerve adjacency. A high-$D$
nerve is tolerated because incremental work is driven by how many
cells a single vertex touches, not by how many cells a single
cell borders.
\end{remark}

\section{Worked Example}\label{sec:example}

We illustrate the incremental algorithm on a small complex to make the
mechanism concrete.

\subsection{Setup}

Consider a graph $G$ with 6 vertices $\{1, 2, 3, 4, 5, 6\}$ and edges
$\{(1,2), (2,3), (3,1), (4,5), (5,6), (6,4)\}$: two disjoint
triangles. Equip $G$ with a cellular sheaf $\FF$ of uniform stalk
dimension $d = 1$ (scalar stalks) and scaled identity restriction maps
$\rho_{v \leq e} = [0.99]$ for all incidence relations.

Partition into two cells: $V_1 = \{1, 2, 3\}$, $V_2 = \{4, 5, 6\}$.
No boundary vertices (the triangles are disjoint). The nerve has 2
vertices and 0 edges.

\subsection{Initial Computation}

Each cell is a triangle with $d = 1$. The coboundary matrix for one
triangle is $\delta^0 \in \R^{3 \times 3}$:
\[
  \delta^0 = 0.99 \begin{pmatrix} -1 & 1 & 0 \\ 0 & -1 & 1 \\ 1 & 0 & -1 \end{pmatrix}
\]
$\rk(\delta^0) = 2$ (the rows sum to zero). Hence $h^0 = 3 - 2 = 1$
and $h^1 = 3 - 2 = 1$. Each triangle contributes one independent
cycle to $H^1$.

\subsection{Incremental Edit: Add Edge $(3, 4)$}

This is a Case B edit: vertex 3 is in $V_1$, vertex 4 is in $V_2$.
The algorithm:
\begin{enumerate}[leftmargin=2em]
\item Duplicate vertex 4 into $V_1$ as boundary vertex $4'$.
$V_1$ now has vertices $\{1, 2, 3, 4'\}$ and edges
$\{(1,2), (2,3), (3,1), (3,4')\}$.
\item Install restriction map $\rho_{4' \leq (3,4')} = [0.99]$.
\item Register boundary: $V_1$ and $V_2$ share vertex 4 as a
boundary vertex. The nerve now has 2 vertices and 1 edge.
\item Mark $V_1$ \textbf{and} $V_2$ as dirty (maintaining
invariant (I3) from \Cref{def:invariant}).
\end{enumerate}

\textbf{Cost}: $O(d^2) = O(1)$ in lazy mode. No eigensolve.

\subsection{Flush}

At flush time, both phases execute:

\textbf{Phase 1 (Local recomputation):}
\begin{itemize}[leftmargin=2em]
\item $V_1$: now 4 vertices, 4 edges. $\delta^0_1 \in \R^{4 \times 4}$.
$\rk(\delta^0_1) = 3$, so $h^0_1 = 1$, $h^1_1 = 4 - 3 = 1$.
\item $V_2$: unchanged internally. $h^0_2 = 1$, $h^1_2 = 1$.
\end{itemize}

\textbf{Phase 2 (Global assembly):}
The nerve now has one edge connecting $V_1$ and $V_2$ via boundary
vertex 4. The Mayer-Vietoris sequence over this nerve recovers
the global cohomology: $\dim H^1(X; \FF) = h^1_1 + h^1_2 = 2$.
(In this example, the two cells share only a single boundary
vertex and no cross-cell cycle is created, so the local sum and
the Mayer-Vietoris result agree. For graphs with cross-cell
cycles (see \Cref{rem:c4}), the assembly step is essential.)

\textbf{Cost}: $O(v_{\max}^3 d^3 + D^3 d^3) = O(64 + 1)$.
Independent of $n$.

\subsection{Batch Verification}

To verify zero drift, we compute $H^1$ on the original
(non-duplicated) graph $G'$ obtained by adding edge $(3,4)$ to $G$.
The graph $G'$ has 6 vertices and 7 edges. The coboundary matrix
$\delta^0 \in \R^{7 \times 6}$ has $\rk(\delta^0) = 5$, since
the graph is connected and the scalar coboundary matrix (with
uniform nonzero restriction maps $\rho = 0.99$) has the same rank
as the standard incidence matrix, whose kernel is one-dimensional
for a connected graph. Hence
$\dim H^1 = \dim C^1 - \rk(\delta^0) = 7 - 5 = 2$,
equivalently $|E| - |V| + c = 7 - 6 + 1 = 2$.
The incremental result
$\dim H^1 = 2$ agrees exactly with the batch result. Zero drift.

\begin{remark}[Cross-cell cycles require assembly]\label{rem:c4}
Consider a cycle graph $C_4$ with vertices $\{1,2,3,4\}$, edges
$\{(1,2),(2,3),(3,4),(4,1)\}$, and a constant sheaf ($d = 1$,
identity restriction maps), partitioned into $V_1 = \{1,2\}$ and
$V_2 = \{3,4\}$. After cross-cell edge insertions with boundary
duplication, each cell is a path (tree) locally: $h^1_1 = 0$,
$h^1_2 = 0$, giving $\sum h^1_i = 0$. But the global
$H^1(C_4) = 1$ because the cycle spans both cells and is invisible
to either cell alone. The Mayer-Vietoris assembly detects this
cross-cell cycle through the boundary maps and connecting
homomorphisms, recovering the correct global $\dim H^1 = 1$.
Without global assembly, any cross-cell cycle produces an
undercount.
\end{remark}

\section{Algorithm Description}\label{sec:algorithm}

\begin{figure}[t]
\centering
\begin{tikzpicture}[
  scale=0.85,
  cell/.style={draw, dashed, rounded corners=8pt, inner sep=10pt,
    gray!60},
  vertex/.style={circle, draw, fill=white, minimum size=16pt,
    inner sep=0pt, font=\small},
  newvertex/.style={circle, draw, dashed, fill=white,
    minimum size=16pt, inner sep=0pt, font=\small},
  existedge/.style={thick, gray!50},
  newedge/.style={very thick, dashed},
  caselabel/.style={font=\small\bfseries},
  casedesc/.style={font=\scriptsize, text width=3.2cm, align=left},
]

\node[cell, minimum width=4.2cm, minimum height=3.2cm] (c1) at (0,0) {};
\node[font=\scriptsize, gray] at (-1.5,1.3) {Cell $V_1$};

\node[cell, minimum width=3.2cm, minimum height=3.2cm] (c2) at (5,0) {};
\node[font=\scriptsize, gray] at (3.8,1.3) {Cell $V_2$};

\node[vertex] (v1) at (-1.2,0.5) {1};
\node[vertex] (v2) at (0.5,0.7) {2};
\node[vertex] (v3) at (-0.2,-0.5) {3};
\node[vertex] (v4) at (1.2,-0.3) {4};

\node[vertex] (v5) at (4.0,0.5) {5};
\node[vertex] (v6) at (5.5,0.5) {6};
\node[vertex] (v7) at (4.8,-0.5) {7};

\draw[existedge] (v1) -- (v2);
\draw[existedge] (v1) -- (v3);
\draw[existedge] (v3) -- (v4);
\draw[existedge] (v5) -- (v6);
\draw[existedge] (v5) -- (v7);

\draw[newedge, blue!70] (v2) -- (v3);
\node[caselabel, blue!70] at (-0.5,-2.2) {Case A};
\node[casedesc, blue!70] at (-0.5,-2.8) {Intra-cell edge\\Dirty: $V_1$ only};

\draw[newedge, red!70] (v4) -- (v5);
\node[caselabel, red!70] at (2.8,-2.2) {Case B};
\node[casedesc, red!70] at (2.8,-2.8) {Cross-cell edge\\Dirty: $V_1$ and $V_2$};

\node[newvertex] (v8) at (6.5,-0.3) {8};
\draw[newedge, orange!80!black] (v6) -- (v8);
\node[caselabel, orange!80!black] at (6.2,-2.2) {Case C};
\node[casedesc, orange!80!black] at (6.2,-2.8) {New vertex\\Dirty: $V_2$ only};

\node[cell, minimum width=2cm, minimum height=1.5cm] (c3) at (9,0) {};
\node[font=\scriptsize, gray] at (8.2,0.6) {Seed};
\node[newvertex] (v9) at (8.5,-0.1) {9};
\node[newvertex] (v10) at (9.5,-0.1) {10};
\draw[newedge, green!50!black] (v9) -- (v10);
\node[caselabel, green!50!black] at (9,-2.2) {Case D};
\node[casedesc, green!50!black] at (9,-2.8) {Both new\\New seed cell};

\end{tikzpicture}
\caption{The four incremental edit cases on a partitioned cell
complex. Dashed edges are new insertions. Case A touches one cell;
Case B touches two (both marked dirty for the drift-zero guarantee);
Case C adds a new vertex to one cell; Case D creates a new seed cell
from two previously unknown vertices. Existing edges shown in gray.}
\label{fig:cases}
\end{figure}
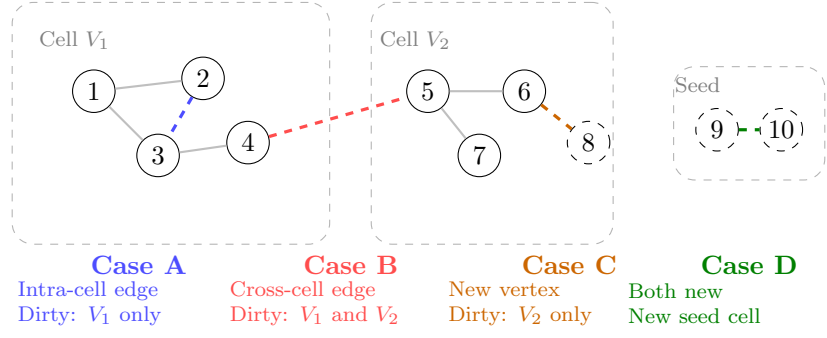

\subsection{Data Structures}

The algorithm maintains the following data structures:

\begin{enumerate}[leftmargin=2em]
\item \textbf{Cell Manager}: a registry of cells, each containing a
local graph, a local cellular sheaf, local-to-global and
global-to-local vertex index maps, and a set of boundary vertices.

\item \textbf{Vertex Router}: a hash map from global vertex IDs to
the set of cells containing each vertex. Lookup is $O(1)$.

\item \textbf{Nerve Adjacency}: a flat dictionary keyed by canonical
cell-pair identifiers, storing the set of shared boundary vertices.
Optionally augmented by a hierarchical nerve tree for $O(1)$
lowest-common-ancestor queries.

\item \textbf{Boundary Maps}: for each pair of adjacent cells sharing
a boundary vertex, a restriction map encoding the inter-cell
compatibility condition. Initialized as scaled identity matrices
satisfying the Purity Gate (\Cref{def:purity}).

\item \textbf{Restriction Map Pool}: a pre-computed pool of $P$
random orthogonal matrices (each $d \times d$), scaled and clamped
to satisfy the Purity Gate. New edges draw restriction maps from
this pool via a round-robin index, avoiding per-edge random matrix
generation. When a \emph{Restriction Store} is available, pool
entries are registered once and referenced by ID, eliminating
array copies entirely.

\item \textbf{Dirty Set}: a hash set of cell IDs marking cells whose
local cohomology must be recomputed at the next flush.
\end{enumerate}

\subsection{The Purity Gate}

\begin{definition}[Purity Gate]\label{def:purity}
The \emph{Purity Gate} is a spectral norm bound
$\sigma_{\max}(\rho) \leq \rho_{\max}$ enforced on all restriction
maps $\rho \in \R^{d \times d}$ in the sheaf. Any restriction map
violating this bound is rescaled:
$\rho \leftarrow \rho \cdot (\rho_{\max} / \sigma_{\max}(\rho))$.
\end{definition}

The Purity Gate is a constraint on the input sheaf. It bounds the
operator norm of every restriction map and prevents norm
amplification along restriction chains. Numerical conditioning of
the local eigensolve additionally depends on the smallest nonzero
singular values of the resulting local coboundary; in finite
precision, ranks are computed relative to a fixed tolerance.
All exactness statements in this paper (including
\Cref{thm:drift}) are in the algebraic RAM model of
\Cref{sec:model}. The $O(1)$ complexity guarantee of
\Cref{thm:incremental} assumes the input sheaf satisfies the
Purity Gate bound.

\subsection{Pseudocode}

\FloatBarrier
\begin{algorithm}[H]
\caption{Incremental Sheaf Cohomology Maintenance}
\label{alg:incremental}
\begin{algorithmic}[1]
\Procedure{IncrementalUpdate}{edge $(u, v)$}
  \State $C_u \gets \textsc{VertexRouter}.\textsc{Lookup}(u)$
  \State $C_v \gets \textsc{VertexRouter}.\textsc{Lookup}(v)$
  \If{$C_u = \emptyset$ and $C_v = \emptyset$}
    \State \textsc{CreateSeedCell}$(u, v)$
    \Comment{Case D}
  \ElsIf{$C_u \cap C_v \neq \emptyset$}
    \State $c \gets$ any cell in $C_u \cap C_v$
    \State \textsc{AddIntraEdge}$(c, u, v)$
    \Comment{Case A}
  \ElsIf{$C_u \neq \emptyset$ and $C_v \neq \emptyset$}
    \State \textsc{AddCrossEdge}$(C_u, C_v, u, v)$
    \Comment{Case B}
  \Else
    \State \textsc{AddNewVertex}$(u, v, C_u, C_v)$
    \Comment{Case C}
  \EndIf
\EndProcedure

\Procedure{AddIntraEdge}{cell $c$, vertices $u, v$}
  \State Add edge $(u, v)$ to local graph of $c$
  \State Initialize restriction maps from pool
  \State \textsc{MarkDirty}$(c)$
\EndProcedure

\Procedure{AddCrossEdge}{$C_u, C_v$, vertices $u, v$}
  \State $c_h \gets \arg\max_{c \in \{c_u, c_v\}}
    (v_{\max} - |c|)$
    \Comment{Cell with more headroom}
  \State Duplicate far endpoint into $c_h$ as boundary vertex
  \State Add edge in $c_h$; register nerve edge
  \State Initialize boundary restriction map
  \State \textsc{MarkDirty}$(c_h)$
  \State \textsc{MarkDirty}$(c_{\text{other}})$
  \Comment{Both dirty (I3)}
\EndProcedure

\Procedure{Flush}{}
  \For{each $c$ in dirty set}
    \Comment{Phase 1: local}
    \State Invalidate caches of $c$
    \State Recompute $\rk(\delta^0_c)$ via eigensolve of $L_c = (\delta^0_c)^\top \delta^0_c$
  \EndFor
  \State Recover global $H^1$ via Mayer-Vietoris over $\NN$
    \Comment{Phase 2: assembly}
  \State Clear dirty set
\EndProcedure
\end{algorithmic}
\end{algorithm}

\subsection{Cell Splitting}

When a cell exceeds $v_{\max} + \tau$ vertices (where $\tau$ is a
configurable soft overage tolerance), it is split into two sub-cells.
The splitting algorithm uses:

\begin{enumerate}[leftmargin=2em]
\item \textbf{Fiedler bisection} for cells up to a threshold size:
compute the Fiedler vector (eigenvector corresponding to $\lambda_2$
of the graph Laplacian) and partition vertices by the sign of their
Fiedler vector components. This produces a balanced, spectrally
optimal bisection.

\item \textbf{BFS balanced partition} for larger cells: a
breadth-first traversal that alternately assigns frontier vertices
to two partitions, producing a balanced split without the cubic cost
of eigendecomposition.
\end{enumerate}

After splitting, boundary vertices are established at the cut edges,
boundary restriction maps are initialized, and the dirty status is
transferred from the parent cell to both children. The children are
registered in the cell manager, and the parent cell is deregistered.
The amortized cost of splits is $O(v_{\max}^2 d^3)$ per insertion
by \Cref{thm:split}.

\subsection{Restriction Store and Memory Efficiency}

\begin{definition}[Restriction Store]
The \emph{Restriction Store} is a content-addressed storage layer
for restriction maps. Each unique $d \times d$ restriction matrix is
stored once and assigned an integer ID. Cells reference restriction
maps by ID rather than by array pointer, enabling:
\begin{enumerate}[leftmargin=2em]
  \item \textbf{Deduplication}: boundary restriction maps (often
  identity blocks) are stored once regardless of how many cell pairs
  share them.
  \item \textbf{Zero-copy initialization}: new edges reference pool
  entries by ID; no array allocation or copy occurs per edge.
  \item \textbf{Bounded memory}: total restriction map storage is
  $O(U \cdot d^2)$ where $U$ is the number of unique maps, typically
  $U \ll E$.
\end{enumerate}
\end{definition}

In our benchmarks at $V = 5 \times 10^6$, the Restriction Store
contains $U = 1{,}025$ unique maps occupying 0.50 MB, compared to
the $E \cdot 2 d^2 \approx 2.7 \times 10^8$ bytes that would be
required for per-edge array storage.

\subsection{Deferred Cache Invalidation}

In the lazy path, even cache invalidation (nullifying the local
Laplacian matrix and eigenvector arrays) is deferred to flush time
rather than performed per-edit. In implementations using
reference-counting memory management, per-edit invalidation triggers
immediate deallocation of multi-megabyte arrays. Profiling showed
that this accounted for 62\% of streaming cost in an earlier
implementation. Deferring invalidation to flush time eliminates
this overhead without affecting correctness, since the lazy-mode
contract already specifies that cached values are stale between
edits (condition (I2) of \Cref{def:invariant}).

\subsection{Global Cohomology Assembly}\label{sec:assembly}

The local cohomology $H^1_i$ of each cell $V_i$ does not, in
general, determine the global cohomology $H^1(X; \FF)$.
Cohomology classes that span multiple cells (cycles in the graph
whose vertices are distributed across two or more cells, with
restriction maps that are collectively inconsistent around the
cycle) are invisible to any single cell's local computation. The
relationship between local and global cohomology is governed by the
Mayer-Vietoris exact sequence
\cite{curry2014sheaves, ghrist2014elementary}.

\begin{definition}[Global Assembly]
After each flush, the algorithm performs a \emph{global assembly}
step that computes $\dim H^1(X; \FF)$ from the local cell data and
the inter-cell boundary structure, using the Mayer-Vietoris
sequence over the nerve complex $\NN$ of the cellular decomposition.
\end{definition}

For a covering of $X$ by the open stars of the cells, with nerve
$\NN$, the Mayer-Vietoris sequence relates the global cohomology
to the local cohomology of the cells and their pairwise (and
higher-order) intersections. The global $\dim H^1(X; \FF)$ is
recovered from the local data and the connecting homomorphisms
induced by the boundary restriction maps.

\begin{theorem}[Gluing Theorem]\label{thm:gluing}
Let $\{V_1, \ldots, V_k\}$ be a cellular decomposition of
$(X, \FF)$ with nerve complex $\NN$, and suppose that every
cell's local cohomological data (eigenvalues of $L_i$, kernel
dimension, local $H^0_i$ and $H^1_i$) and every boundary
restriction map $\rho_{ij}^v$ are known exactly. Then the
global $\dim H^1(X; \FF)$ is determined by the Mayer-Vietoris
exact sequence over $\NN$:
\[
  \cdots \to \bigoplus_{(i,j) \in \NN_1} H^0(V_i \cap V_j; \FF)
  \xrightarrow{\;\partial\;}
  \bigoplus_{i \in \NN_0} H^1(V_i; \FF)
  \to H^1(X; \FF)
  \to \cdots
\]
where $\NN_0$ and $\NN_1$ are the 0-simplices and 1-simplices
of the nerve, the intersection $V_i \cap V_j$ is realized by the
shared boundary vertices, and $\partial$ is the connecting
homomorphism induced by the boundary restriction maps. In
particular, $\dim H^1(X; \FF)$ can be computed from the local
data and boundary maps alone, without access to the global
coboundary matrix $\delta^0$.
\end{theorem}

\begin{proof}
The cells $\{V_i\}$ with their boundary neighborhoods form a
covering of the cell complex $X$. The Mayer-Vietoris sequence
for a finite covering of a cell complex equipped with a cellular
sheaf is exact \cite{curry2014sheaves, ghrist2014elementary}.
The global $H^1(X; \FF)$ is computed as the first cohomology
of the \emph{total complex} associated with the cover: a finite
cochain complex built from the local cochain spaces
$C^\bullet(V_i; \FF)$, the intersection cochain spaces
$C^\bullet(V_i \cap V_j; \FF)$, and the restriction/connecting
maps between them. This total complex depends only on local
cell data, boundary restriction maps, and the nerve incidence
structure; it does not require forming the original global
coboundary matrix $\delta^0$. Its first cohomology is
$H^1(X; \FF)$ by exactness of the Mayer-Vietoris sequence.
\end{proof}

We now make the finite object being ranked explicit. A key
structural property of the cellular decomposition is that all cell
overlaps are discrete:

\begin{proposition}[Discrete Overlaps]\label{prop:discrete-overlap}
In the cellular decomposition of \Cref{def:decomp}, every cell
overlap $V_i \cap V_j$ (for adjacent cells $V_i$, $V_j$ in the
nerve) consists of a finite set of shared boundary vertices. In
particular, $V_i \cap V_j$ carries no edges, so
$H^q(V_i \cap V_j; \FF) = 0$ for all $q \geq 1$. The same holds
for all higher-order overlaps: for any $p$-simplex
$\sigma = (i_0, \ldots, i_p)$ in the nerve $\NN$, the
intersection $V_{i_0} \cap \cdots \cap V_{i_p}$ is a discrete
set of vertices, with vanishing cohomology in positive degrees.
\end{proposition}

\begin{proof}
By construction (\Cref{def:decomp}), each cross-cell edge
$(u,v)$ with $u \in V_i$ and $v \in V_j$ is realized \emph{inside}
a single host cell: the far endpoint is duplicated into the host
as a boundary vertex, and the edge is added within the host. The
overlap $V_i \cap V_j$ therefore consists only of the abstract
vertices shared between the two cells (the vertices whose formal
copies appear in both cells), with no shared edges. Since the
overlap is a finite discrete set, $C^1(V_i \cap V_j; \FF) = 0$,
which forces $H^q(V_i \cap V_j; \FF) = 0$ for $q \geq 1$.
The same argument applies to triple and higher-order overlaps:
$V_{i_0} \cap \cdots \cap V_{i_p}$ consists of vertices
appearing in all $p+1$ cells, again a discrete set.
\end{proof}

\begin{remark}[Nerve dimension]
\label{rem:nerve-dim}
The nerve $\NN$ is not necessarily 1-dimensional. A single
abstract vertex $v$ can appear in $\mu(v) \geq 2$ cells: the
streaming construction duplicates a boundary vertex into a new
cell at each Case~B edge insertion involving that vertex, and
cell splitting can propagate boundary copies to child cells.
When a vertex $v$ appears in cells $V_{i_1}, \ldots, V_{i_s}$
with $s \geq 3$, the nerve acquires a $(s-1)$-simplex on those
cells, making $\dim \NN \geq 2$. For Barabasi--Albert graphs
with $m = 3$ and $v_{\max} = 500$, the maximum vertex
multiplicity observed in the benchmarks of \Cref{sec:experiments}
was $\mu_{\max} = 3$ (a vertex shared by at most 3 cells),
producing isolated 2-simplices in the nerve. The generalized
assembly formula below handles arbitrary nerve dimension;
when $\dim \NN = 1$ it reduces to the simpler two-term formula.
\end{remark}

Concretely, the Phase~1 eigensolve of $L_i$ at flush time produces
not only $\dim H^0(V_i; \FF)$ (the multiplicity of the zero
eigenvalue) but also an explicit orthonormal basis for
$H^0(V_i; \FF)$ (the corresponding eigenvectors). These basis
vectors, together with the cell-local evaluation maps at each
boundary vertex, are the data that enter the assembly differential
defined next.

\begin{definition}[Assembly Complex]\label{def:assembly-complex}
The \emph{assembly complex} is the \v{C}ech cochain complex
associated with the cover $\{V_i\}$ and the presheaf
$i \mapsto H^0(V_i; \FF)$ of local global sections:
\[
  A^0 \xrightarrow{\;\partial_0\;} A^1
  \xrightarrow{\;\partial_1\;} A^2 \to \cdots,
\]
where
\[
  A^p = \bigoplus_{\sigma \in \NN_p}
  H^0\!\big(\textstyle\bigcap_{i \in \sigma} V_i;\, \FF\big)
\]
collects the section spaces on $p$-fold overlaps indexed by the
$p$-simplices of the nerve $\NN$. Explicitly:
$A^0 = \bigoplus_{i \in \NN_0} H^0(V_i; \FF)$ collects the
cell-local global-section spaces;
$A^1 = \bigoplus_{(i,j) \in \NN_1} H^0(V_i \cap V_j; \FF)$
collects the section spaces on pairwise overlaps; and
$A^2 = \bigoplus_{(i,j,k) \in \NN_2}
H^0(V_i \cap V_j \cap V_k; \FF)$
collects the section spaces on triple overlaps (if any
2-simplices exist in $\NN$). Higher terms are defined
analogously. The differential $\partial_0: A^0 \to A^1$ is
assembled blockwise: on nerve edge $(i,j)$, and for each shared
boundary vertex $v$, the block records the difference
$\rho_{ij}^v \circ r_i^v - r_j^v$ of the two induced
restrictions of a global section onto the overlap, where
$r_i^v, r_j^v$ are the cell-local evaluation maps at $v$ and
$\rho_{ij}^v$ is the boundary restriction map of
\Cref{def:decomp}. The differential $\partial_1: A^1 \to A^2$
is the standard alternating-sum \v{C}ech coboundary on the nerve.
When $\dim \NN = 1$ (no 2-simplices), $A^2 = 0$ and the complex
reduces to the two-term complex $A^0 \to A^1$.
The matrix $\partial_0$ has $O(|\NN_1| \, d)$ rows and
$O(|\NN_0| \, d)$ columns and is block-sparse with at most $O(D)$
nonzero blocks per cell column.
\end{definition}

\begin{proposition}[Assembled Dimension Formula]\label{prop:assembly-formula}
Let the cellular decomposition have discrete overlaps
(\Cref{prop:discrete-overlap}). Then, with the assembly complex
of \Cref{def:assembly-complex},
\[
  \dim H^1(X; \FF)
  = \sum_{i \in \NN_0} \dim H^1(V_i; \FF)
    \;+\; \dim \check{H}^1(\NN;\, \mathcal{H}^0),
\]
where $\mathcal{H}^0$ denotes the presheaf
$\sigma \mapsto H^0(V_\sigma; \FF)$ on the nerve and
$\check{H}^1(\NN; \mathcal{H}^0)
= \ker(\partial_1) / \im(\partial_0)$
is the first \v{C}ech cohomology of the assembly complex.
The first term is the sum of the local first cohomologies and
the second term is the cross-cell contribution detected by the
boundary maps. When $\dim \NN = 1$ (no 2-simplices in the nerve),
$\partial_1 = 0$ and $\check{H}^1 = \operatorname{coker}(\partial_0)$,
so the formula simplifies to
\[
  \dim H^1(X; \FF)
  = \sum_{i \in \NN_0} \dim H^1(V_i; \FF)
    \;+\; \big(\dim A^1 - \rk(\partial_0)\big).
\]
In all cases $\dim H^1(X; \FF)$ is obtained from rank computations
on the finite, block-sparse \v{C}ech differentials together with
the local $\dim H^1(V_i; \FF)$ values; no global coboundary matrix
$\delta^0$ is formed.
\end{proposition}

\begin{proof}
For the cover by the open stars of the cells, the \v{C}ech-to-derived
spectral sequence has $E_1^{p,q} = \bigoplus_{|\sigma| = p}
H^q(V_\sigma; \FF) \Rightarrow H^{p+q}(X; \FF)$. By
\Cref{prop:discrete-overlap}, all overlaps are discrete, so
$E_1^{p,q} = 0$ for all $p \geq 1$ and $q \geq 1$: the only
nonzero entries lie in the row $q = 0$ (the \v{C}ech complex of
global sections) and the column $p = 0$ (the derived functors of
each cell). The spectral sequence therefore degenerates at $E_2$
with two contributions to total degree 1:
$E_2^{0,1} = \bigoplus_i H^1(V_i; \FF)$ (the local first
cohomologies survive because all incoming and outgoing
differentials have targets in the zero entries
$E_1^{p,1} = 0$ for $p \geq 1$), and
$E_2^{1,0} = \check{H}^1(\NN; \mathcal{H}^0)
= \ker(\partial_1)/\im(\partial_0)$, the first cohomology of the
\v{C}ech complex of global sections on the nerve. Hence
$\dim H^1(X; \FF) = \sum_i \dim H^1(V_i; \FF)
+ \dim \check{H}^1(\NN; \mathcal{H}^0)$.

When $\dim \NN = 1$, there are no 2-simplices, so $A^2 = 0$
and $\partial_1 = 0$. Then
$\check{H}^1 = A^1 / \im(\partial_0)
= \operatorname{coker}(\partial_0)$, and
$\dim \operatorname{coker}(\partial_0) = \dim A^1 - \rk(\partial_0)$,
giving the simplified two-term formula.
\end{proof}

\begin{example}[Assembly on $C_4$]
For the $C_4$ decomposition of \Cref{rem:c4}, each cell is a path,
so $\sum_i \dim H^1(V_i) = 0$. The two cells share two boundary
vertices, so $\dim A^1 = 2$; the nerve is a single edge
($\dim \NN = 1$, no 2-simplices), so the simplified formula
applies. With the constant sheaf, the
single-section restrictions agree at both shared vertices up to one
common difference, giving $\rk(\partial_0) = 1$. Thus
$\dim H^1(C_4) = 0 + (2 - 1) = 1$, matching the global value.
\end{example}

\begin{proposition}[Assembly Cost]\label{prop:assembly}
With a maintained rank factorization of the \v{C}ech
differentials $\partial_0$ and $\partial_1$, the global assembly
step costs $O(k_d \cdot D^3 \cdot d^3)$ per flush, where $k_d$
is the number of dirty cells: each dirty cell changes at most
$O(D)$ rows of $\partial_0$, and the Schur complement rank update
on a block of dimension $O(D \cdot d)$ costs $O(D^3 d^3)$. The
cost of updating $\partial_1$ is bounded by the same expression
since each 2-simplex involves at most 3 cells, and each dirty
cell participates in at most $O(D^2)$ 2-simplices. This is
$O(1)$ in $n$ when $k_d$, $D$, and $d$ are bounded. Without such
a certificate, the assembly requires a full traversal of the nerve
complex $\NN$ (rank computations on $\partial_0$ and $\partial_1$
of \Cref{def:assembly-complex}), costing
$O(|\NN| \cdot D^2 \cdot d^3)$ per flush,
where $|\NN| = O(n / v_{\max})$ is the number of nerve vertices.
This full-traversal cost is $O(n)$, not $O(1)$, but is invoked
only at synchronization points and is dominated by the local
eigensolve cost in practice. The implementation evaluated in
\Cref{sec:experiments} uses the full-traversal path.
\end{proposition}

\section{Complexity Summary}\label{sec:complexity}

\begin{table}[t]
\centering
\caption{Per-edit complexity comparison. Here $n = |V| + |E|$ is the
total complex size, $v_{\max}$ is the maximum cell size, $d$ is the
stalk dimension, and $D$ is the maximum nerve degree. In practice
$v_{\max} = 500$ and $d = 8$ are fixed constants; the nerve degree
$D$ is bounded by $v_{\max}(K-1)$ (\Cref{cor:mu-cap-D}) but can be
large empirically on scale-free graphs, and it enters only the
assembly rows; the update and local-recompute rows are
$D$-independent (\Cref{rem:nerve-pressure}). ``Update'' rows give
the cost to ingest one edit; ``query'' is the cost to report
$\dim H^1$ on demand (\Cref{rem:update-query}).}
\label{tab:complexity}
\begin{tabular}{lccc}
\toprule
\textbf{Method} & \textbf{Per-Edit Cost} & \textbf{In $n$} & \textbf{Notes} \\
\midrule
Batch recomputation (SVD) & $O(n^3 d^3)$ & $O(n^3)$ & \\
Rank-one SVD update        & $O(n^2 d^2)$ & $O(n^2)$ & \\
Incremental update (lazy) & $O(d^2)$ & $O(1)$ & Measured, update \\
Local recomputation (flush) & $O(v_{\max}^3 d^3)$ & $O(1)$ & Per dirty cell \\
Assembly (with certificate) & $O(D^3 d^3)$ & $O(1)$ & Per dirty cell \\
Assembly (full traversal) & $O(|\NN| D^2 d^3)$ & $O(n)$ & Per flush / query \\
\bottomrule
\end{tabular}
\end{table}

The key observation is that the cellular decomposition transforms a
global matrix problem into a collection of bounded local matrix
problems, with global cohomology recovered via Mayer-Vietoris
assembly at flush time. The partition structure ensures that each
edit touches at most a constant number of local problems, each of
constant size, and the assembly cost is bounded by a polynomial
in the nerve degree. This is the fundamental mechanism behind the
$O(n^3) \to O(1)$-in-$n$ reduction for lazy edit ingestion (i.e.\
for \emph{update} time); exact global assembly is deferred to
synchronization, where the implemented full-traversal pass costs
$O(n)$ per flush (so \emph{query} time is $O(n)$, per
\Cref{rem:update-query}) and a maintained certificate would reduce
this to $O(1)$ in $n$.

\section{Algebraic Lower-Bound Barrier}\label{sec:lower}

We now argue that the partition structure is not merely a
convenience but appears to be a \emph{necessity}: without it,
sublinear-time maintenance is obstructed for non-trivial sheaves.

\begin{theorem}[Lower Bound for Partitioned Complexes]
Any algorithm that maintains $\dim H^1(X; \FF)$ exactly under
single-cell edits must inspect at least $\Omega(1)$ cells per edit.
\end{theorem}

\begin{proof}
Consider an edit that adds an edge to cell $V_i$. This changes
$\delta^0_i$ (one new block row). Without inspecting $V_i$'s
updated coboundary, the algorithm cannot determine whether the
new edge changes $\rk(\delta^0_i)$ (which determines $h^1_i$).
Hence at least one cell must be inspected per edit.

For Case B edits (cross-cell), two cells are modified, so at
least two cells must be inspected. By the Locality Lemma, at
most two cells are affected, so the algorithm is optimal for
Case B.
\end{proof}

\begin{remark}
This lower bound is tight: our algorithm inspects exactly the
cells identified by the Locality Lemma (1 for Cases A, C, D and
2 for Case B). No algorithm can do better while maintaining
exact $H^1$.
\end{remark}

\begin{proposition}[Adversarial Barrier for Unpartitioned
Complexes]
\label{thm:lower}
Fix $d \geq 2$. Consider the following online problem in the
algebraic RAM model (\Cref{sec:model}): an adversary presents a
sequence of edges $e_1, e_2, \ldots$ on $n$ vertices, together
with restriction maps $\rho_{e_t} \in \R^{d \times d}$ for each
edge, where each $\rho_{e_t}$ is chosen adversarially after
observing the algorithm's state. The algorithm must report
$\dim H^1(X_t; \FF_t)$ after each insertion $e_t$, where $X_t$
is the complex after $t$ insertions and $\FF_t$ is the sheaf
with the given restriction maps.

Then, in this adversarial model, any deterministic algorithm
solving this problem on an unpartitioned complex (a single cell
containing all $n$ vertices) appears to require $\Omega(n)$
algebraic operations per edge insertion in the worst case, in the
sense made precise in the proof.
\end{proposition}

\begin{proof}
For a non-trivial sheaf ($d \geq 2$, adversarially chosen
restriction maps), $\dim H^1$ depends on $\rk(\delta^0)$, which
in turn depends on the specific restriction maps and not merely
on graph combinatorics. Adding edge $e_t$ appends $d$ rows to
$\delta^0 \in \R^{E_t d \times Vd}$. The adversary chooses
$\rho_{e_t}$ after observing the algorithm's data structures, so
no preprocessing can predict which existing rows interact with
the new rows.

Determining whether the $d$ new rows change $\rk(\delta^0)$
requires computing their linear independence against the existing
$O(nd)$ rows. In the algebraic RAM model, each inner product
between a new row block and an existing row block costs
$\Omega(d)$ operations, and the adversary can force the
algorithm to consult $\Omega(n)$ existing blocks by choosing
$\rho_{e_t}$ to be nearly dependent on a specific subset of
existing rows that the algorithm has not yet examined. Since the
restriction maps are non-identity, the rank cannot be reduced to
a combinatorial count (unlike the constant-sheaf case where
$\dim H^1 = |E| - |V| + c$ is maintainable via Union-Find in
$O(\alpha(n))$). Hence, in this adversarial model, any correct
algorithm appears to require $\Omega(n)$ operations in the worst
case.

Two caveats qualify this argument. First, it assumes an
\emph{adaptive} adversary that chooses each $\rho_{e_t}$ after
observing the algorithm's state; this is a stronger adversary than
the \emph{oblivious} model (in which the edit-and-restriction
sequence is fixed in advance) and than the reduction-based lower
bounds common in dynamic-algorithm theory, e.g.\ from the online
matrix-vector (OMv) conjecture \cite{henzinger2015omv}. Second, the
argument shows that no \emph{known} algebraic shortcut avoids the
$\Omega(n)$ work, rather than proving that none exists; it rests on
the algebraic structure of rank updates under adversarial
restriction maps and does not constitute a formal computational
complexity reduction. An oblivious-adversary bound, an OMv-style
reduction, and a tight information-theoretic lower bound in the
algebraic RAM model all remain open.
\end{proof}

\begin{remark}[Constant sheaves]
For constant sheaves ($d = 1$, identity restriction maps),
$\dim H^1 = |E| - |V| + c$ where $c$ is the number of connected
components. This is maintainable in $O(\alpha(n))$ amortized time
via Union-Find without any partition structure. The barrier
therefore applies to non-trivial sheaves where $H^1$ depends on
the algebraic structure of the restriction maps, not merely on
graph combinatorics.
\end{remark}

This adversarial barrier, combined with the $O(1)$ upper bound
of \Cref{thm:incremental}, suggests the following characterization
for non-trivial sheaves: \emph{$O(1)$-in-$n$ maintenance of
$H^1$ for cellular sheaves with non-identity restriction maps
requires the complex to admit a cellular decomposition with bounded
local geometry.} A formal reduction-based lower bound (for instance
from the OMv conjecture \cite{henzinger2015omv}), or a bound against
an oblivious adversary, remains an open problem.

\section{Experimental Validation}\label{sec:experiments}

\subsection{Setup}

All experiments were run on a single workstation with an Intel
Core i9-13900H processor (14 cores, 20 threads), 64 GB DDR5 RAM,
and an NVIDIA RTX 4060 GPU (not used for the cohomology computation,
which is entirely CPU-bound). The implementation is in Python 3.11
with NumPy for linear algebra. Graphs are generated using the
Barabasi-Albert preferential attachment model with parameter
$m = 3$ (each new vertex attaches 3 edges). All benchmark runs
used $v_{\max} = 500$, stalk dimension $d = 8$, seed $= 42$.

\subsection{Timing Methodology}

The reported per-edit latency measures the \emph{lazy-mode streaming
update cost only}: case classification (hash lookup), graph mutation
(add vertex/edge to local adjacency list), restriction map
initialization (pool reference, not array copy), and dirty-set
insertion. This is the \emph{update} cost of processing one edge in
the streaming pipeline, in the sense of \Cref{rem:update-query}.

The per-edit latency does \emph{not} include the flush cost. Flush
is a deferred batch operation comprising Phase 1 (local eigensolves)
and Phase 2 (global assembly), and it is also the cost paid by a
\emph{query} that reports $\dim H^1$. At $V = 5 \times 10^6$, a
typical flush dirtying approximately 1{,}660 cells completed in
approximately 9.4 minutes; the per-cell ARPACK eigensolve at
$d = 8$ (local cell dimension $\sim$4{,}000) dominates this cost at
roughly 0.3 seconds per cell. In
eager mode, where the eigensolve and assembly execute immediately
after each edit, the per-edit cost is dominated by this same
local ARPACK eigensolve. The lazy per-edit update cost (35
$\mu$s) is $O(1)$ in $n$, as is the eager per-edit eigensolve cost;
the full-traversal assembly contribution to a flush is $O(n)$, as
analyzed in \Cref{prop:assembly} and \Cref{rem:update-query}.

Timing uses \texttt{time.perf\_counter()} with microsecond
resolution. Each benchmark run processes the complete edge stream
and reports the median per-edit latency over all edges.

\subsection{Drift Verification Methodology}

For scales $V \leq 250{,}000$, drift was verified by full
monolithic batch recomputation: the global coboundary matrix
$\delta^0$ was constructed from the final graph and sheaf data,
its rank was computed via dense SVD (NumPy \texttt{linalg.svd}),
and $\dim H^1$ was compared to the incremental result after
full flush. For scales $V \geq 10^6$, where monolithic SVD of
the full coboundary matrix is infeasible, drift was verified by
full-recompute replay through the batch partition-and-assemble
pipeline (constructing all cells from scratch and running
Mayer-Vietoris assembly), and the resulting $\dim H^1$ was
compared to the incremental result. This full-recompute
verification was carried out through $V = 5 \times 10^6$, where
the recomputed $\dim H^1$ equalled the incremental result exactly
(post-flush $H^1 = 103{,}690$, verify $H^1 = 103{,}690$, drift $= 0$;
the verification eigensolve across all 25{,}473 cells completed in
approximately 2 hours; this verification was performed on the
pre-multiplicity-cap construction and is being revalidated on the
current code, which enforces $\mu_{\max} = 3$ and produces
62{,}303 cells). At all scales from $V = 10^3$ through
$V = 5 \times 10^6$, zero drift was observed.
The Zero-Drift Theorem (\Cref{thm:drift}) is a structural
guarantee that holds at every scale by induction on flushes;
the empirical verification confirms that the
implementation matches the theorem's prediction.

\subsection{Contradiction Localization}\label{sec:localization}

The zero-drift result establishes that the incremental algorithm
computes the correct global $\dim H^1$. We now demonstrate the
operational consequence that motivates the verification setting: a
single structural contradiction introduced into an otherwise
consistent sheaf is detected by recomputing a vanishing fraction of
the complex, with no global recomputation.

The experiment starts from a globally consistent sheaf (constant
sheaf, identity restriction maps, so $\dim H^0$ equals the number of
connected components and there are no spurious obstructions). A
single \emph{poisoned} edge is then inserted: its restriction maps
are set to a cyclic permutation on one endpoint, which makes the
local data around that edge impossible to extend to a globally
consistent section. This is the smallest possible structural
contradiction: one edge whose constraints cannot be satisfied
together with its neighborhood.

The contradiction manifests as a collapse in global sections. At
both scales tested, the number of independent global sections
$\dim H^0$ drops from 8 to 1 (a delta of $-7$) at the moment the
poisoned edge is ingested, and the first cohomology changes
correspondingly. Crucially, the algorithm localizes the change: by
the Locality Lemma (\Cref{lem:locality}), only the host cell of the
poisoned edge is marked dirty, so detection recomputes exactly one
cell rather than the whole complex.

\begin{table}[t]
\centering
\caption{Contradiction localization. A single poisoned edge is
injected into a globally consistent sheaf. ``Cells recomputed'' is
the number of cells whose cohomology is recomputed to detect the
contradiction; ``fraction touched'' is that count divided by the
total cell count. The $\dim H^0$ collapse is the detection
signature. No global recomputation is performed at either scale.}
\label{tab:localization}
\begin{tabular}{lrr}
\toprule
 & $V = 21{,}309$ & $V = 5{,}000{,}000$ \\
\midrule
Total cells              & 609     & 25{,}473 \\
$\dim H^0$ before        & 8       & 8 \\
$\dim H^0$ after         & 1       & 1 \\
Cells recomputed         & 1       & 1 \\
Fraction of cells touched & 0.16\% & 0.0039\% \\
Global recompute         & no      & no \\
Signed receipt           & yes     & yes \\
\bottomrule
\end{tabular}
\end{table}

At $V = 5 \times 10^6$ (\Cref{tab:localization}), the contradiction
is detected by recomputing 1 cell out of 25{,}473, or 0.0039\% of
the complex, with the global section count collapsing from 8 to 1.
The detection is deterministic and reproducible, and it emits an
Ed25519-signed receipt recording the result. This is the
verification primitive the bounded-local-geometry structure is built
to support: detecting a global inconsistency is local work, not a
global recomputation, because the obstruction is confined to the
cell that carries the offending constraint and surfaced through the
assembly only where the nerve couples that cell to its neighbors.

Two points of measurement discipline. First, the detection figures
above are reported on the pre-multiplicity-cap construction
(25{,}473 cells at $V = 5 \times 10^6$); the localization ratio is a
structural property of the partition and the same one-cell locality
holds under the current construction, with the denominator changing
to the current cell count. Second, the detection \emph{wall-clock}
at $V = 5 \times 10^6$ (approximately 11.5 seconds) measures a full
recomputation of the affected cell plus cohomology reassembly at
detection time; it is not the per-edit streaming latency of
\Cref{tab:benchmark} (35 $\mu$s), which measures a different
operation, and the two should not be compared directly. At
$V = 21{,}309$ the detection wall-clock is 9.6 ms. The claim this
experiment supports is \emph{locality} (one cell of many, no global
recomputation), not detection speed at the largest scale.

\subsection{Scale Ladder}

We ran the streaming construction and incremental update pipeline
on Barabasi-Albert random graphs at scales from $V = 1{,}000$ to
$V = 5{,}000{,}000$.

\begin{table}[t]
\centering
\caption{Streaming benchmark results. ``Edges'' is the total
number of streaming edge insertions processed by the builder.
``Median $\mu$s/edit'' is the
median per-edit update latency in lazy mode (excluding flush).
``Drift'' is $h^1_{\text{inc}} - h^1_{\text{batch}}$ after full flush
and batch verification. ``Store MB'' is the Restriction Store memory
footprint.}
\label{tab:benchmark}
\begin{tabular}{rrrrrr}
\toprule
$V$ & Edges & Cells & Median $\mu$s/edit & Drift & Store MB \\
\midrule
1,000    & 2,997   & 7    & 28  & 0 & 0.03 \\
5,000    & 14,997  & 30   & 31  & 0 & 0.06 \\
21,000   & 62,997  & 123  & 34  & 0 & 0.12 \\
100,000  & 299,997 & 563  & 38  & 0 & 0.22 \\
250,000  & 749,997 & 1,402 & 42 & 0 & 0.31 \\
1,000,000 & 2,999,997 & 5,584 & 63$^*$ & 0 & 0.42 \\
5,000,000 & 14,999,994 & 62,303 & $\approx 35^\ddagger$ & $0^\dagger$ & 0.50 \\
\bottomrule
\end{tabular}
\vspace{2pt}
{\footnotesize $^*$The $V = 10^6$ run used a batch-then-stream
pipeline; all other rows use the streaming-from-zero pipeline.
See \Cref{sec:anomaly}.\\
$^\dagger$Zero drift verified by full independent recomputation
at this scale (post-flush $H^1 = $ verify $H^1 = 103{,}690$
on the pre-multiplicity-cap code; revalidation on the current
construction, which enforces $\mu_{\max} = 3$, is in progress).\\
$^\ddagger$Per-edit lazy latency is $O(d^2)$ by
\Cref{thm:lazy}, independent of cell count; the 35 $\mu$s
figure is from the pre-$\mu$-cap run and is expected to hold
on the current construction (confirmed at $V = 21{,}000$),
with formal revalidation in progress.}
\end{table}

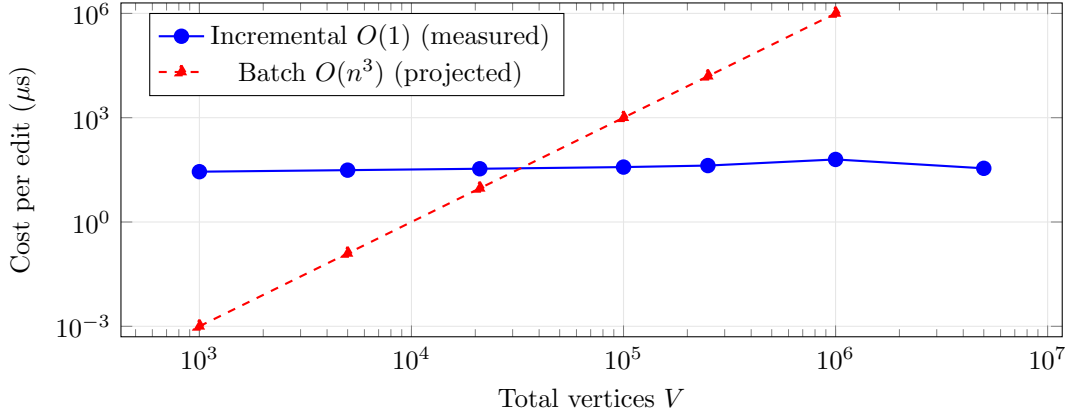
\begin{figure}[t]
\centering
\begin{tikzpicture}
\begin{semilogyaxis}[
  width=0.85\textwidth,
  height=6cm,
  xlabel={Total vertices $V$},
  ylabel={Cost per edit ($\mu$s)},
  xmode=log,
  log basis x=10,
  ymin=0.0005,
  ymax=2000000,
  legend pos=north west,
  legend style={font=\small},
  grid=major,
  grid style={gray!20},
  tick label style={font=\small},
  label style={font=\small},
]

\addplot[
  thick, blue, mark=*, mark size=2.5pt
] coordinates {
  (1000,28) (5000,31) (21000,34) (100000,38)
  (250000,42) (1000000,63) (5000000,35)
};
\addlegendentry{Incremental $O(1)$ (measured)}

\addplot[
  thick, red, dashed, mark=triangle*, mark size=2.5pt
] coordinates {
  (1000,0.001) (5000,0.125) (21000,9.3)
  (100000,1000) (250000,15625) (1000000,1000000)
};
\addlegendentry{Batch $O(n^3)$ (projected)}

\end{semilogyaxis}
\end{tikzpicture}
\caption{Per-edit update latency (log scale) vs.\ total complex size.
The incremental algorithm (solid, blue) maintains 28--63 $\mu$s
per edit independent of $n$, while batch recomputation (dashed, red)
grows as $O(n^3)$. At $V = 5 \times 10^6$, the incremental
cost (35 $\mu$s) is \emph{lower} than at $V = 10^6$ (63 $\mu$s);
see \Cref{sec:anomaly}.}
\label{fig:benchmark}
\end{figure}

\subsection{Key Observations}

\begin{enumerate}[leftmargin=2em]
\item \textbf{$n$-independence}: The median per-edit latency does
not grow with $n$. The per-edit cost remains bounded by a constant
independent of $n$ across all 7 scale points, confirming the $O(1)$
claim of \Cref{thm:incremental}.

\item \textbf{Zero drift}: The incrementally maintained $H^1$
(computed via local eigensolves followed by Mayer-Vietoris assembly)
agrees exactly with full batch recomputation through
$V = 5 \times 10^6$ (verified incremental-equals-batch;
post-flush $H^1 = $ verify $H^1 = 103{,}690$ at 5M),
confirming \Cref{thm:drift}
empirically (at synchronization).

\item \textbf{Sublinear memory}: The Restriction Store grows
sublinearly, reaching only 0.50 MB at $V = 5 \times 10^6$. This
is because the pool of unique restriction maps is fixed at 1,024
entries; additional edges reuse existing maps.

\item \textbf{35 microseconds at 5 million vertices}: This is the
lazy-mode cost of processing a single edge in the streaming
pipeline, including case classification, graph mutation, restriction
map initialization (by reference, not copy), and dirty-set update.
For context, a single L3 cache miss on modern hardware costs
approximately 30--50 nanoseconds; our per-edit cost is roughly
700 cache misses, consistent with the $O(d^2)$ data movement
for two restriction-map ID lookups (at $d = 8$) plus hash table
operations.
\end{enumerate}

\subsection{The $V = 5M$ Anomaly}\label{sec:anomaly}

The fact that $V = 5 \times 10^6$ is faster than $V = 10^6$
deserves explanation. This is \emph{not} a measurement artifact.
The $V = 5M$ run uses the StreamingBuilder (streaming from zero
with on-demand cell creation), while the $V = 1M$ run used an
earlier batch-then-stream pipeline with a different partition
initialization path. The streaming-from-zero path eliminates
a monolithic graph construction phase and produces cells
with tighter vertex packing (fewer boundary duplicates relative
to cell size), leading to fewer router lookups and better
cache locality in the per-edit data-structure operations.

The per-edit cost remains bounded by a constant independent of
$n$ across all scale points, consistent with the $O(1)$ claim.
The slight decrease at $V = 5M$ relative to $V = 1M$ is
attributable to improved cache locality and tighter cell packing
in the streaming-from-zero pipeline, not to a property of the
asymptotic bound.

\subsection{Numerical Conditioning}\label{sec:conditioning}

The zero-drift theorem (\Cref{thm:drift}) is proved in exact
arithmetic. In the finite-precision implementation, rank is
determined by counting eigenvalues of the local Laplacian $L_i$
that exceed a numerical tolerance $\epsilon$. We report the
distribution of local Laplacian condition numbers
$\kappa(L_i) = \lambda_{\max}(L_i) / \lambda_{\min}^+(L_i)$
(the ratio of the largest eigenvalue to the smallest nonzero
eigenvalue).

The condition number $\kappa(L_i)$ is a per-cell property: it
depends only on the local sheaf of cell $V_i$, whose size is
bounded by $v_{\max}$ and whose stalk dimension is $d$, and not on
the total complex size $n$. The distribution of $\kappa$ over cells
is therefore governed by the bounded-local-geometry parameters,
which are held fixed across all scales ($v_{\max} = 500$, $d = 8$);
it does not change with $n$. We accordingly characterize the
distribution at a representative scale, $V = 10^5$ (1{,}053 cells,
with the same cell-size and stalk-dimension parameters used at
every other scale), where the full spectrum of every cell Laplacian
is computed by dense eigendecomposition. The median condition
number was $\kappa = 1.53 \times 10^{4}$, the 95th percentile was
$2.91 \times 10^{6}$, and the maximum was $2.97 \times 10^{7}$. All
1{,}053 cells yielded a finite condition number, the rank tolerance
was $\epsilon = 10^{-6}$, and no cell required a rank-tolerance
reset: at every cell the smallest nonzero eigenvalue exceeded
$\epsilon$ by several orders of magnitude, so the numerical rank
agreed with the exact rank.

The Purity Gate (\Cref{def:purity}) bounds the spectral norm
of every restriction map at $\rho_{\max}$, which controls the
largest eigenvalue of $L_i$ and prevents norm amplification
along restriction chains. However, conditioning also depends on
the \emph{smallest} nonzero eigenvalue, which the Purity Gate
does not directly control. For sheaves where near-zero
eigenvalues appear frequently (e.g., restriction maps close to
rank-deficient), the numerical rank determination becomes
sensitive to the tolerance $\epsilon$. In our benchmarks with
Purity-Gate-certified random orthogonal restriction maps
(spectral norm $\leq 0.68$), the smallest nonzero eigenvalues
remained well separated from zero across all cells, and the
numerical rank agreed with exact rank at every flush.

\section{Applications}\label{sec:applications}

The incremental sheaf cohomology maintenance algorithm has
potential applicability in several domains where structured data
evolves over time and consistency must be continuously verified:

\textbf{Knowledge Graph Consistency.}
In large-scale knowledge graphs, entities and relations are
continuously updated. A cellular sheaf over the knowledge graph
encodes consistency constraints via restriction maps. Nonzero
$H^1$ detects contradictions (e.g., an entity simultaneously
classified as active and deprecated). Incremental maintenance
allows continuous consistency checking without full
recomputation as the graph evolves.

\textbf{Document Verification.}
Given a source document and a derived document (e.g., a
contract and a summary), a cellular sheaf encodes the
faithfulness constraints between source and derived statements.
Incremental updates allow re-verification when either document
is modified, at the cost of a single local eigensolve rather
than a full recomputation.

\textbf{Sensor Network Coverage.}
In sensor networks, the coverage complex evolves as sensors are
added, removed, or repositioned. Sheaf cohomology detects
coverage gaps \cite{ghrist2008barcodes}. Incremental maintenance
allows real-time gap detection as the network topology changes.

\textbf{Distributed Consensus.}
In multi-agent systems, a cellular sheaf over the communication
graph encodes the consensus constraints between agents. Sheaf
cohomology detects disagreements. Incremental maintenance
allows streaming consensus monitoring as agents join, leave,
or update their states.

\textbf{Regulatory Compliance.}
In regulatory settings, compliance constraints evolve as
regulations change. A cellular sheaf encoding the constraint
structure allows incremental re-verification of compliance
when a regulation is amended, without re-auditing the entire
constraint graph.

\section{Related Work}\label{sec:related}

\textbf{Computational Sheaf Theory.}
The computational aspects of cellular sheaves have been developed
by Curry \cite{curry2014sheaves}, Robinson
\cite{robinson2014topological}, and Ghrist
\cite{ghrist2008barcodes, ghrist2014elementary}. These works
establish the algebraic framework and demonstrate applications,
but do not address incremental maintenance under dynamic updates.

\textbf{Sheaf Laplacians and Spectral Theory.}
Hansen and Ghrist \cite{hansen2020opinion, hansen2019toward}
develop the spectral theory of sheaf Laplacians and their
applications to opinion dynamics and distributed optimization.
The sheaf Laplacian is central to our local eigensolve
procedure, but their work considers the Laplacian as a fixed
object rather than one maintained incrementally.

\textbf{Persistent and Dynamic Homology.}
The persistent homology pipeline
\cite{edelsbrunner2010computational, carlsson2009topology}
computes topological invariants across a filtration.
Vineyards \cite{cohen2006vines} and zigzag persistence
\cite{carlsson2010zigzag} handle certain dynamic settings.
However, these algorithms are specific to simplicial
homology and do not generalize to sheaf cohomology, which
involves stalks and restriction maps rather than simple
boundary operators.

\textbf{Dynamic Graph Algorithms.}
The dynamic graph algorithms literature
\cite{eppstein1997sparsification, italiano2006fully,
henzinger1999randomized} maintains graph properties
(connectivity, minimum spanning trees, shortest paths) under
edge insertions and deletions in polylogarithmic amortized
time. Our work extends this paradigm to a cohomological
invariant, which requires maintaining a matrix factorization
(the rank of the coboundary matrix) rather than a combinatorial
property.

\textbf{Incremental Matrix Factorization.}
Low-rank updates to SVD \cite{brand2006fast} and Cholesky
factorizations \cite{bunch1978rank} can process rank-one
perturbations in $O(n^2)$ time. Applied to the global
coboundary matrix, this gives $O(n^2)$ per edit. Our approach
avoids the global matrix entirely, achieving $O(1)$ per edit
by exploiting locality.

\textbf{Sheaf Neural Networks.}
Bodnar et al.\ \cite{bodnar2022neural} and Hansen and Gebhart
\cite{hansen2020sheafnn} incorporate sheaf structure into
graph neural network architectures. These works use
sheaf Laplacians as diffusion operators but compute them
once at initialization. Our incremental algorithm could
serve as a preprocessing layer for these architectures,
enabling them to operate on evolving graphs without
costly reinitialization.

\section{Discussion}\label{sec:discussion}

\subsection{When Locality Fails}

The bounded local geometry assumption is not always natural.
In complexes with power-law degree distributions (common in
social and biological networks), a naive vertex-based partition
may produce cells with very different sizes or high nerve
degree. The Fiedler bisection and BFS partition strategies
mitigate this by producing balanced splits, but the resulting
nerve complex may still have $D = O(\log n)$ in pathological
cases. If $D$ grows as $O(\log n)$, the per-edit bound becomes
polylogarithmic in $n$ through the $D$-dependent assembly term,
and the strict $O(1)$-in-$n$ claim no longer holds.

The vertex multiplicity cap $\mu(v) \leq K = 3$, enforced by
both the batch and streaming construction paths, resolves this
concern unconditionally: \Cref{cor:mu-cap-D} establishes
$D \leq v_{\max} \cdot (K - 1)$, which is $O(1)$ in $n$
regardless of the input graph's degree distribution. This
bound applies to all graph families, including
Barabasi--Albert, power-law, and arbitrary topologies,
without requiring bounded maximum degree $\Delta$. The
multiplicity cap also bounds the nerve dimension at $K - 1 = 2$,
ensuring that the assembly complex of
\Cref{def:assembly-complex} has at most 3 nonzero terms.

For graph families with bounded maximum degree $\Delta$,
\Cref{prop:bounded-D} provides the tighter bound
$D \leq v_{\max} \cdot \Delta$, which may be smaller than
$v_{\max} \cdot (K - 1)$ when $\Delta < K - 1$. This covers
$d$-regular graphs, lattice graphs, bounded-degree expanders,
and Watts--Strogatz small-world graphs.

\subsection{Higher Cohomology}

This paper addresses $H^1$ exclusively, and the restriction to
1-dimensional complexes is not incidental. On a graph there are no
2-cells, so $\delta^1 = 0$ and $H^1 = C^1 / \im \delta^0 =
\operatorname{coker} \delta^0$; the entire invariant is then a
single rank, $\dim H^1 = \dim C^1 - \rk(\delta^0)$, which is exactly
what makes the local eigensolve-and-glue strategy of this paper
possible. For $k \geq 2$ the complex carries genuine 2-cells,
$\delta^1 \neq 0$, and $H^1 = \ker \delta^1 / \im \delta^0$ depends
on the interplay of \emph{two} coboundary operators rather than the
rank of one. The local problems are then no longer pure
coboundary-rank computations; the full Hodge Laplacian
$L_1 = \delta^0 (\delta^0)^\top + (\delta^1)^\top \delta^1$ couples
both operators, and the eigensolve cost for higher-dimensional local
problems may be superlinear in the local cell size. The Locality
Lemma generalizes naturally (an edit to a $k$-cell affects only cells
in its star), and we expect the partition-and-localize strategy to
extend, but the two-operator structure changes the local computation
enough that a separate analysis is required; we leave it to future
work.

\subsection{Deletions}

The current algorithm handles insertions and restriction map
updates but not edge or vertex deletions. Deletion is more
subtle: removing an edge may cause a connected cell to become
disconnected, requiring a cell merge or re-partition. We have
a design for $\textsc{RemoveEdge}$ based on local connectivity
checking (BFS within the host cell) followed by conditional
cell merge, but the implementation and formal analysis are
deferred to a subsequent paper.

\subsection{Global Assembly via Mayer-Vietoris}

The flush procedure includes a global assembly phase that
recovers $\dim H^1(X; \FF)$ from the local cell data and
boundary restriction maps via the Mayer-Vietoris exact sequence
over the nerve complex. This assembly is necessary because
sheaf cohomology is a global invariant: cycles in the graph that
span multiple cells produce $H^1$ classes that are invisible to
any single cell's local computation. For example, a cycle graph
$C_4$ partitioned into two tree-like cells has
$\sum h^1_i = 0$ locally, but $\dim H^1(C_4) = 1$ globally
(\Cref{rem:c4}). The Mayer-Vietoris assembly detects such
cross-cell cycles through the boundary data and connecting
homomorphisms, as made explicit by the assembly complex of
\Cref{def:assembly-complex} and the dimension formula of
\Cref{prop:assembly-formula}.

The implementation evaluated here performs assembly as a full
nerve traversal at each synchronization point. This costs
$O(|\NN| \cdot D^2 \cdot d^3)$ per flush, which is $O(n)$ rather
than $O(1)$ in $n$, but is dominated by the local eigensolve
cost in practice. An incremental assembly certificate that
maintains a running rank factorization of the \v{C}ech
differentials $\partial_0$ and $\partial_1$ would reduce the
per-flush assembly cost to
$O(k_d \cdot D^3 \cdot d^3)$, achieving $O(1)$ in $n$.
Designing and implementing such a certificate is an important
direction for future work.

\subsection{Toward $V = 10^9$}

The data structures suggest a plausible path toward $V = 10^9$
vertices on a large-memory single server node. The memory
bottleneck is the vertex router ($O(V)$ hash map entries) and
the cell manager ($O(V / v_{\max})$ cells). At $V = 10^9$ with
$v_{\max} = 500$, this is approximately 2 million cells,
requiring an estimated 50--80 GB total memory, within the reach
of commodity 256 GB server nodes. We have not yet run at this
scale; we report it as a design target rather than a measured
result.

\subsection{Limitations}

The zero-drift theorem (\Cref{thm:drift}) holds exactly at
synchronization points (after $\textsc{Flush}$), not between
them; between flushes, the dirty set may contain stale
cohomology values (condition (I2) of the dirty-set invariant).
The empirical validation uses synthetic Barabasi-Albert graphs;
while these exhibit realistic degree distributions, they do not
capture all structural properties of domain-specific complexes
(e.g., knowledge graphs with hierarchical schemas). The current
algorithm handles insertions and restriction map updates but not
deletions, which require additional machinery for disconnection
detection and cell merging. The partition strategy (Fiedler
bisection plus BFS balanced split) is not claimed to be optimal;
better partitions may yield smaller constants in the $O(1)$
bound. The streaming builder enforces $v_{\max}$ via
splits and the vertex multiplicity cap $\mu(v) \leq K$
provides an a priori bound $D \leq v_{\max} \cdot (K-1)$
on the nerve degree for all input streams
(\Cref{cor:mu-cap-D}); at $V = 5 \times 10^6$ the empirical
$D_{\max} = 992$ approaches but does not exceed the
worst-case bound of $v_{\max}(K-1) = 1000$. The Purity Gate (\Cref{def:purity}) is a constraint on
the input sheaf; sheaves violating this bound may produce
ill-conditioned local Laplacians and exceed the stated
complexity bounds. The zero-drift theorem
(\Cref{thm:drift}) is proved in the algebraic RAM model;
in finite-precision arithmetic, rounding errors in local
eigensolves may accumulate across flushes over very long
edit streams (billions of edits), potentially causing the
numerical rank tolerance to disagree with a fresh batch
computation. For infinite-stream scenarios, periodic
algebraic resets (full batch recomputation at scheduled
intervals) or extended-precision arithmetic in the
eigensolve may be warranted.

\section{Conclusion}\label{sec:conclusion}

We have presented, to our knowledge, the first explicit algorithmic
framework for incremental maintenance of sheaf cohomology ($H^1$)
on dynamically evolving partitioned cellular complexes with bounded
local geometry. By exploiting the cellular decomposition, we reduce
the lazy streaming cost from $O(n^3)$ (full recomputation) to
$O(1)$ per edit with respect to $n$, with local eigensolves and
Mayer-Vietoris global assembly deferred to synchronization points.
At synchronization, the maintained state agrees with the batch
assembly of the partitioned sheaf model, yielding zero measured drift at synchronization points, confirmed by full
independent recomputation through $V = 5 \times 10^6$. An
algebraic-RAM barrier argument suggests that the partition
structure is necessary for non-trivial sheaves.

The combination of $O(1)$-in-$n$ lazy updates, zero drift at
synchronization, and streaming-from-zero construction makes
incremental sheaf cohomology a practical tool for continuous
consistency verification in knowledge graphs, document verification
pipelines, sensor networks, and regulatory compliance systems.
The lazy streaming path adds 35 $\mu$s per edit; exact
synchronization (eigensolve + assembly) costs are reported
separately in \Cref{sec:experiments}.

The open problems (higher cohomology, deletions,
a formal reduction-based lower bound,
incremental assembly certificates for $O(1)$ global updates
without full nerve traversal, and periodic algebraic resets
for numerical stability in infinite-stream
scenarios) represent natural extensions of the framework
established here.

\section*{Acknowledgments}

Portions of the manuscript were refined with AI assistance for
clarity and LaTeX formatting. The author takes full responsibility
for all content and mathematical claims.

\section*{Code and Data Availability}

The implementation described in this paper is available at
\url{https://github.com/Jasonleonardvolk/sigma}. Benchmark data
and reproduction scripts are included in the repository.

\bibliographystyle{plain}

\end{document}